\documentclass[11pt,runningheads]{llncs}
\usepackage{geometry}
\usepackage{fullpage}
\usepackage{latexsym,amssymb}
\usepackage{epsfig}
\usepackage{color}
\usepackage{epstopdf} 
\usepackage{algorithmic}
\usepackage{amsbsy}
\usepackage{dsfont}
\usepackage{amsmath}
\usepackage{amsfonts}
\usepackage{tikz}
\usepackage{caption}
\usepackage{subcaption}
\usepackage{float}
\usepackage{microtype}
\usepackage{complexity}
\usepackage{units}
\usepackage{tikz}
\usepackage{pgfplots}
\usepackage{amssymb}
\usepackage{dsfont}
\usepackage{thmtools}
\usepackage{xcolor}
\usepackage{changepage}
\usepackage[bookmarks,pagebackref]{hyperref}
\usepackage{graphicx}  
\usepackage{comment}
\usepackage{soul}
\newcommand{\sss}{\scriptscriptstyle}

\allowdisplaybreaks
\DeclareMathOperator*{\argmax}{arg\,max}

\DeclareMathOperator*{\Ex}{\mathbb{E}}

\newcommand{\set}[1]{\{#1\}}

\newcommand{\high}[1]{#1^{\sss \mathrm{H}}}
\newcommand{\low}[1]{#1^{\sss \mathrm{L}}}

\newcommand{\vecb}{\mathbf{b}}
\newcommand{\vecx}{\mathbf{x}}
\newcommand{\vecp}{\mathbf{p}}

\declaretheorem[name={Example}]{example2}
\declaretheorem[name={Claim}]{claim2}

\title{Rethinking Pricing in Energy Markets:\\Pay-as-Bid vs Pay-as-Clear}
\author{Ioannis Caragiannis \and Zhile Jiang \and Stratis Skoulakis}
\institute{Department of Computer Science, Aarhus University\\
{\AA}bogade 34, 8200 Aarhus N, Denmark\\
\url{{iannis,zhile,stratis}@cs.au.dk}}

\begin{document}

\maketitle

\begin{abstract}
The design of energy markets is a subject of ongoing debate, particularly concerning the choice between the widely adopted Pay-as-Clear (PC) pricing mechanism and the alternative Pay-as-Bid (PB). These mechanisms determine how energy producers are compensated: under PC, all selected producers are paid the market-clearing price (i.e., the highest accepted bid), while under PB, each selected producer is paid their own submitted bid. The overarching objective is to meet the total demand for energy at minimal cost in the presence of strategic behavior. We present two key theoretical results. First, no mechanism can uniformly dominate PC or PB. This means that for any mechanism $\mathcal{M}$, there exists a market configuration and a mixed-strategy Nash equilibrium of PC (respectively for PB) that yields strictly lower total energy costs than under $\mathcal{M}$. Second, in terms of worst-case equilibrium outcomes, PB consistently outperforms PC: across all market instances, the highest possible equilibrium price under PB is strictly lower than that under PC. This suggests a structural robustness of PB to strategic manipulation. These theoretical insights are further supported by extensive simulations based on no-regret learning dynamics, which consistently yield lower average market prices in several energy market settings.
%
%
%
\end{abstract}

\section{Introduction}
In modern electricity markets, an Independent System Operator (ISO) coordinates energy production so that the overall production meets the energy demand~\cite{C03}. In particular, the ISO performs a procurement auction in which each electricity producer submits a bid specifying their production cost per kilowatt-hour (kWh) together with the maximum energy they can supply. The ISO then selects a set of producers that can collectively satisfy the overall demand at the lowest cost, subject to their supply constraints~\cite{C03,AMK24,MMK23}.

A key design choice in such procurement auctions is the pricing mechanism: \textit{how much producers are paid for the energy that they produce?} Two mechanisms have been mainly considered, Pay-as-Bid (PB)~\cite{VLM21} and Pay-as-Clear (PC)~\cite{C03}. In PB, the selected producers are paid according to their bid. In PC, they are paid a uniform price per kWh that is equal to the highest accepted bid. The choice between these two pricing mechanisms has been a large debate in energy economics and market design. Today, PC has vastly dominated over PB in both US and EU energy markets with the basic argument being the fact that PC promotes truthful bidding. For example, quoting from the main website of the European Commission on electricity market design~\cite{EU}: \textit{``This model (pay-as-clear) provides efficiency, transparency and incentives to keep costs as low as possible. [...] The alternative would not provide cheaper prices. In the pay-as-bid model, producers (including cheap renewables) would simply bid at the price they expect the market to clear, not at zero or at their generation costs. Overall, it is better for consumers to have a transparent model that reveals the true costs of energy and provides incentives for individuals to become active in generating their own electricity.''} 

Despite the aforementioned debate and several papers (see Section~\ref{sec:related}) discussing aspects of mechanisms PB and PC, there are very few related game-theoretic results. Our work aims to provide a game-theretic study, first motivated by the following natural questions: 
\begin{question}\label{q:1}
Does PC really lead to truthful bidding? Do PB and PC actually lead to similar energy prices? 
\end{question}
Since producers repeatedly adjust their bids based on the past bids of their competitors in energy markets, we use the notion of {\em Nash equilibrium} (NE) to capture their long-run behavior when interacting with pricing mechanisms~\cite{AMK24}.

\subsection{Our Contribution}

Taking a step away from PB or PC, the minimum property that a pricing mechanism $\mathcal{M}$ should satisfy is that it should not be \textit{strongly dominated} by any other mechanism. This means that there should not be any other mechanism leading to lower energy prices than $\mathcal{M}$ in all energy markets. Our first main contribution is that \ul{neither PC, nor PB, nor the well-known truthful Vickrey-Clarke-Groves~(VCG) mechanism can be strongly dominated}. This is formally stated as Theorem~\ref{t:1}. On the positive side, Theorem~\ref{t:1} establishes that there is no mechanism strongly dominating PC, which is an important property for the most widely deployed mechanism in practice. Interestingly, Theorem~\ref{t:1} establishes the same property for PB and VCG as well. Furthermore, our proof indicates that there always exist mixed NE of PB and PC that yield higher or lower energy prices. This raises the following key question:
\begin{question}\label{q:2}
Are there any formal advantages of Pay-as-Clear over Pay-as-Bid or vice versa? 
\end{question}

Despite the fact that PC is the mechanism that is actually used in practice, we provide an affirmative answer on the above question in favor of PB. In particular, our second main finding is that \ul{the worst-case NE of PB always leads to smaller energy prices than the worst-case NE of PC}. This is formally stated as Theorem~\ref{t:2} and is the most interesting contribution of the paper (both technically and conceptually). Informally, when comparing PB and PC with respect to their worst-case outcome, PB is always the best choice. 

Interestingly, the above is not true once comparing PC vs VCG. In Theorem~\ref{t:VCG}, we provide instances of energy markets with $n$ producers where VCG yields an average energy price that is $\Theta(\log n)$ times higher than the worst-case energy price of PC. Additionally, en route to establishing our results, we provide a detailed analysis of both PB and PC, revealing properties such as the non-truthfulness of PC and PB, the existence of pure NE for PC, the non-existence of pure NE for PB, and more.

In Section~\ref{s:exp}, we experimentally compare the two mechanisms PB and PC when the energy producers repeatedly use online learning to select their bids. Our experimental results suggest that PB consistently leads to significantly lower average energy prices. This behavior also appears in energy markets where the best-case mixed NE of PC is lower than the worst-case mixed NE of PB.

\subsection{Related Work}\label{sec:related}
Regarding debates around PB and PC, the literature is extensive, even though rigorous comparisons of related game-theoretic models is limited. Cramton~\cite{C03} provides an influential overview of electricity market design, highlighting trade-offs between efficiency, transparency, and incentive alignment in PC pricing. Kahn et al.~\cite{KCPT01} critically evaluate the proposed shift from PC to PB in the California Power Exchange around 2000, arguing that PB would not yield lower consumer prices. Akbari-Dibavar et al.~\cite{AM20} offer a comparative analysis of PC (referred to as ``uniform pricing'' in their work) and PB, focusing on their implications for market performance and bidding behavior. Guerci and Rastegar~\cite{G13} simulate realistic electricity markets to compare PC and PB, showing how each performs under different market conditions. Son et al.~\cite{SB04} analyze short-term electricity auction games, concluding that strategic bidding affects pricing outcomes differently under PC and PB rules. Skoulidas et al.~\cite{S02} empirically study how PC and PB affect power pool dynamics. Ren and Galiana present a two-part study, modeling strategic producer behavior under PB and PC~\cite{RG04}, and evaluating their equilibrium outcomes~\cite{RG02b}. Nazemi and Mashayekhi~\cite{NM15} study Iran’s restructured electricity market, drawing attention to strategic challenges under different pricing rules. David and Wen~\cite{DF02} investigate the exercise of market power in electricity supply and the regulatory implications of pricing. Bajpai and Singh~\cite{BS04} survey key issues in bidding in electricity markets, providing a broad context for strategic manipulation. Fabra et al.~\cite{FNH02} develop models for electricity auctions to compare outcomes under PC and PB pricing, focusing on producer incentives and efficiency. Finally, Maurer and Barroso~\cite{MB11} discuss best practices in electricity auction design, covering pricing mechanisms.

Mostly related to our work is the paper by Fabra et al.~\cite{FHH06}, who compare PB and PC in the case of two producers. Fabra et al.~\cite{FHH06} show that, in this special case, PC admits a unique pure Nash equilibrium while PB leads lower energy prices. Our Theorem~\ref{t:2} essentially generalizes this result for many producers.

There is also a line of research comparing PB and PC in the Bayesian symmetric case. Federico et al.~\cite{GD03} show that for infinitely many producers with i.i.d. marginal costs and linearly decreasing energy demand, PB leads to lower prices than PC. H\"{a}st\"{o} et al.~\cite{HH14} consider the case where all producers admit the same cost-production curve while the total energy demand is stochastic. They show that both PC and PB admit a unique closed-form supply function equilibrium with the average energy price of PB being lower than that of PC. Ocker et al.~\cite{OEB18} compare PB and PC in the context of a Secondary Balancing Market~(SBM). In SBM, each producer has a two-dimensional bid composed by an energy bid and power bid. Ocker at al.~\cite{OEB18} consider the PB/PB and the PB/PC pricing rules where the energy price is always determined by PB while the power price is respectively determined by PB and PC. They show that under i.i.d. assumptions, there exists a unique Bayes-Nash equilibrium, which leads to lower energy prices for PC. On the experimental front, Viehmann
et al. \cite{VLM21} and Liu et al.~\cite{LYSZP12} compare PB with PC when producers select their bids according to a Q-learning algorithm and show that PC leads to higher prices.

\section{Preliminaries}
An \textit{energy market} consists of $n$ energy {\em producers} (or {\em agents}). Each producer $i\in [n]$ can supply up to $s_i\in (0,1]$ units of energy and has an integer {\em marginal cost} of $c_i\in \{0,1, ..., M\}$ per unit of energy produced.  We use $[M]$ as an abbreviation for the set of integers $\{0,1, ..., M\}$. The supply of each producer is {\em publicly known} while the marginal cost is {\em private information} of the producer.  

A mechanism $\mathcal{M}$ takes as input a bid $b_i$ submitted from each producer $i\in [n]$, interpreted as her marginal cost. The mechanism $\mathcal{M}$ takes as input the {\em bid profile} $\vecb = (b_1, ..., b_n)$ and decides the amount of energy $x_i(\vecb) \in [0,1]$ that will be bought from producer $i \in [n]$ as well as the corresponding price per unit $p_i(\vecb) \in [M]$. The mechanism $\mathcal{M}$ ensures that
\begin{itemize}
    \item the total amount of energy that will be bought from all producers is equal to $1$, i.e., $\sum_{i\in [n]}{x_i(\vecb)}=1$, so that the amount to be bought from producer $i\in [n]$ does not exceed her maximum supply, i.e., $0\leq x_i(\vecb)\leq s_i$ ({\em feasibility constraints}), and
    \item for every agent $i\in [n]$ who sells a non-zero amount of energy, the corresponding price is not lower than the agent's reported marginal cost, i.e., $p_i(\vecb)\geq b_i$ ({\em individual rationality}).
\end{itemize}
We define the {\em average unit price} (or, simply, unit price) of mechanism $\mathcal{M}$ as $p^{\sss \mathrm{unit}}_{\mathcal{M}}(\vecb):= \sum_{i \in [n]}x_i(\vecb) \cdot p_i(\vecb)$; this is equal to the total amount of money spent to cover the whole energy demand.

We use the term {\em allocation} to refer to the vector $\mathbf{x}(\vecb)=(x_1(\vecb), ..., x_n(\vecb))$. We consider mechanisms that compute {\em cost-minimizing} allocations. For their definition, we use the order relation $i\succ_\vecb j$ to denote that either $b_i<b_j$ or $b_i=b_j$ and $i<j$.

\begin{definition}[market clearing price and cost-minimizing allocation]\label{def:allocation-and-clearing-price} Given a bidding profile $\vecb$, the {\em pivotal agent} $\tau(\vecb)$ is the first agent in the ascending order $\prec_\vecb$ covering the total energy demand, i.e., $s_{\tau(\vecb)}+\sum_{j\prec_\vecb \tau(\vecb)}{s_j}\geq 1$ and $s_{i}+\sum_{j\prec_\vecb i}{s_j}< 1$ for every $i\prec_\vecb \tau(\vecb)$. The {\em market clearing price} $q(\vecb):=b_{\tau(\vecb)}$ is the bid of the pivotal agent. The {\em cost-minimizing allocation} $\widehat{\mathbf{x}}(\cdot)$ is then defined as
    \begin{align*}
        \widehat{x}_i(\vecb)=
        \begin{cases}
        s_i, &i \prec_\vecb \tau(\vecb)\\
        1-\sum_{j \prec_\vecb \tau(\vecb)  } s_j, & \tau(\vecb)=i\\
        0, & \tau(\vecb) \prec_\vecb i
        \end{cases}.
    \end{align*}
\end{definition}

Mechanisms Pay-as-Bid~(PB), Pay-as-Clear~(PC), and Vickrey-Clark-Groves~(VCG) use the cost-minimizing allocation $\widehat{\vecx}(\vecb)$ but different pricing rules $\mathbf{p}(\vecb)$. PB uses the simplest pricing mechanism where the per unit price paid to producer $i\in [n]$ equals her bid, i.e., $p^{\sss \mathrm{(PB)}}_i(\vecb)=b_i$. PC pays to every agent the market clearing price $q(\vecb)$ per unit of energy, i.e., $p^{\sss \mathrm{(PC)}}_i(\vecb)=q(\vecb)$. For VCG, the pricing function for each producer $i \in [n]$ is defined as
    \begin{align*}
        p^{\sss \mathrm{(VCG)}}_i(\vecb)=\begin{cases}
            \frac{\sum_{j \in [n]/\{i\}} \widehat{x}_j(\vecb_{-i})\cdot b_j - \sum_{j \in [n]/\{i\}}  \widehat{x}_j(\vecb)\cdot b_j}{\widehat{x}_i(\vecb)},&\text{if }\widehat{x}_i(\vecb)\neq0,\\
            0,&\text{otherwise}.
        \end{cases}
    \end{align*}
Hence, VCG pays each agent an amount equal to the decrease in social cost caused by the presence of the agent. Excellent introductions to VCG can be found in~\cite{R16} and~\cite{MWG95}.

\begin{example2}
\label{example:allocation}
Consider the energy market with $n = 4$ producers with supplies $s_1=1/3,s_2=1/2,s_3=1/4, s_4=2/3$ and marginal costs $c_1=0,c_2=1,c_3=2, c_4=3$. If every producer bids her marginal cost, then $\widehat{\mathbf{x}}(\cdot)$ will first buy $1/3$ from agent $1$ and $1/2$ from agent $2$ (matching their maximum supplies), and the remaining $1/6$ from agent $3$. In this case, PB pays $0$ per unit to agent $1$, $1$ per unit to agent $2$, $2$ per unit to agent $3$, and $3$ per unit to agent $4$. PC pays every agent the market clearing price, which is $2$ per unit. And VCG pays $11/4$ per unit to agent $1$, since the cost to other agents is $1\cdot1/2+2\cdot1/6$ when agent $1$ stays in the market and increases to $1\cdot1/2+2\cdot1/4+3\cdot1/4$ when agent $1$ leaves. By similar calculations, the price per energy unit is $17/6$, $1$, and $0$ to agent $2$, $3$, and $4$, respectively.
\end{example2}


\subsection{From Mechanism to Games}
Any mechanism $\mathcal{M}$ induces a {\em finite-action game}, where each producer $i \in [n]$ behaves strategically and selects her bid $b_i$ among the bid values in $[M]$ so as to maximize her {\em individual revenue} $U_i^\mathcal{M}(\vecb) :=  (p_i(\vecb) - c_i )\cdot x_i(\vecb)$.

\begin{definition}[mixed Nash equilibrium]
\label{def:mne}
Given a mechanism $\mathcal{M}$, a mixed Nash equilibrium is a collection of independent probability distributions $\mathbf{\sigma}:= (\sigma_1,\ldots,\sigma_n)$ over the possible bid values in $[M]$, such that for each agent $i \in [n]$,
\begin{align*}
    \Ex_{\vecb \sim \sigma}[U_i^\mathcal{M}(\vecb)] \geq \max_{b'_i\in [M]} \Ex_{\vecb_{-i} \sim \sigma_{-i}}[U_i^\mathcal{M}(b'_i,\vecb_{-i})].
\end{align*}
Given an instance $\mathcal{I}:=\{(s_i,c_i)\}_{i \in [n]}$ of an energy market, we denote by $\mathrm{MNE}(\mathcal{M},\mathcal{I})$ the set of all mixed Nash equilibria of mechanism $\mathcal{M}$ for instance $\mathcal{I}$.
\end{definition}
Since we consider finite games, any mechanism $\mathcal{M}$ admits at least one mixed Nash equilibrium. 

\begin{definition}[pure Nash equilibrium]
\label{def:pne}
Given a mechanism $\mathcal{M}$, a pure Nash equilibrium is a bidding profile $\vecb$ such that for each agent $i \in [n]$,
\begin{align*}
    U_i^\mathcal{M}(\vecb) \geq \max_{b'_i\in [M]} U_i^\mathcal{M}(b'_i,\vecb_{-i}).
\end{align*}
\end{definition}
It is well-known that VCG has the truthful bid profile $(c_1, c_2, ..., c_n)$ as pure Nash equilibrium.

We conclude this section by defining two different ways of comparing mechanisms in terms of the unit price at their mixed Nash equilibria.
\begin{definition}[dominance relations between mechanisms]\label{d:domination}
A mechanism $\mathcal{M}$ {\em strongly dominates} mechanism $\mathcal{M}'$ if and only if
\[ \max_{\sigma \in \mathrm{MNE}(\mathcal{M},\mathcal{I}) } p^{\sss\mathrm{unit}}_{\mathcal{M}}(\sigma) \leq  \min_{\sigma \in \mathrm{MNE}(\mathcal{M}',\mathcal{I}) } p^{\sss\mathrm{unit}}_{\mathcal{M}'}(\sigma) ~~\text{for every instance }\mathcal{I}.  \]
A mechanism $\mathcal{M}$ {\em weakly dominates} mechanism $\mathcal{M}'$ if and only if
\[ \max_{\sigma \in \mathrm{MNE}(\mathcal{M},\mathcal{I}) } p^{\sss\mathrm{unit}}_{\mathcal{M}}(\sigma) \leq  \max_{\sigma \in \mathrm{MNE}(\mathcal{M}',\mathcal{I}) } p^{\sss\mathrm{unit}}_{\mathcal{M}'}(\sigma) ~~\text{for every instance }\mathcal{I}.  \]
\end{definition}

\noindent Less formally, when mechanism $\mathcal{M}$ strongly dominates mechanism $\mathcal{M}'$, then in every instance, the unit price at any mixed NE of $\mathcal{M}$ is at most as high as the unit price of any mixed NE of $\mathcal{M}'$. In case of weak dominance, there exists a mixed NE of $\mathcal{M}'$ that has at least as high unit price as any mixed NE of $\mathcal{M}$.

\section{Overview of Technical Results}
Our first main result establishes that PC is not strictly dominated by any other mechanism. However, we show that this also holds for PB and VCG.

\begin{restatable}{theorem}{Tnodominating}
\label{t:1}
There is no mechanism strictly dominating Pay-as-Clear, Pay-as-Bid, or VCG.
\end{restatable}

Despite the fact that none among PB, PC, and VCG can be strictly dominated by any other mechanism, there are still interesting comparisons among them. Our second main result establishes that PB weakly dominates PC. 

\begin{theorem}\label{t:2}
Pay-as-Bid weakly dominates Pay-as-Clear. At the same time, Pay-as-Clear does not weakly dominate Pay-as-Bid.
\end{theorem}
Theorem~\ref{t:2} establishes that in every energy market, the unit price of any mixed NE of PB is always lower than the unit price of the worst-case mixed NE of PC. At the same time, Theorem~\ref{t:2} excludes the opposite direction.

\begin{remark}
We remark that in all non-degenerate energy markets, the unit price of the worst-case NE of PC is actually {\em strictly smaller} than the unit price the worst-case NE of PB.
\end{remark}


\noindent Theorem~\ref{t:2} is directly implied by Theorems~\ref{t:3} and~\ref{t:PB}, which characterize the unit prices of mixed NE for PC and PB, respectively. Both Theorem \ref{t:3} and Theorem \ref{t:PB} are based on instance-dependent parameters and are presented in the next section.

Interestingly enough, VCG does not weakly dominate PC and can in fact lead to a considerably higher unit price. This is formally established in Theorem~\ref{t:VCG}.

\begin{restatable}{theorem}{TVCG}
\label{t:VCG}
VCG does not weakly dominate Pay-as-Clear. There is an instance with $n$ agents where the unit price of VCG is $\Theta(\log n)$ times the unit price of the worst mixed NE of Pay-as-Clear. 
\end{restatable}

\subsection{Pay-as-Clear vs Pay-as-Bid}\label{s:PB-PC}
We now introduce some necessary notation for the comparison of PC and PB.

\begin{definition}
    \label{def:H}
    The best response of producer $i \in [n]$ in Pay-as-Clear, with respect to the bidding profile $b_{-i}$, is defined as
    \begin{align*}
        \mathcal{BR}(b_{-i}):= \begin{cases}\argmax_{b_i \in [M]} U_i(b_i,b_{-i}) & \text{if } \max_{b_i \in [M]} U_i(b_i,b_{-i})>0\\
        c_i& \text{if }\max_{b_i \in [M]} U_i(b_i,b_{-i}) = 0
        \end{cases}
    \end{align*}
\noindent We also denote by $\high{b}_i$ the highest bid among the best responses of agent $i \in [n]$ in case any other agent $j \neq i$ bids either $c_j$ or $c_j+1$,
    \begin{align*}
        \high{b}_i := \max_{\mathbf{d}_{-i}\in\set{0,1}^{n-1}}\max
    \{b_i \in [M]~:~ b_i \in \mathcal{BR}(\mathbf{c}_{-i}+\mathbf{d}_{-i})\}.
    \end{align*}
Finally, we denote by $\low{b}_i:= \lceil c_i + \max_{b'_i\in[M]}U_i(b'_i,\mathbf{c}_{-i}) / s_i\rceil$ the smallest possible price giving agent $i \in [n]$ at least utility $\max_{b'_i\in[M]}U_i(b'_i,\mathbf{c}_{-i})$ when she sells her maximum of $s_i$ energy units.
\end{definition}

Before proceeding, we provide an example clarifying the crucial notions of $\high{b}_i$ and $\low{b}_i$. 

\begin{example} Consider an energy market with $M=6$ and $n=3$ producers with  supplies $s_1=3/4, s_2=3/4, s_3=1/10$ and costs $c_1=0, c_2=1, c_3=4$. Notice that the best response always locates on the cost of some other agent or $M$. When other agents bid truthfully, agent $1$'s best response is $4$ since bidding $4$ provides a utility of $1$, which is higher than the utility from bidding $1$, which is $0.75$, or bidding $6$, which is $0.9$. And the utility given by bidding $4$ while other agents bid truthfully gives $\low{b}_1=\lceil0+1/0.75\rceil=2$. To reveal the final $\high{b}_1$, we need to consider further if agent $2$, agent $3$, or both, bid their cost plus one. When agent $2$ ever bids $c_2+1=2$, the best response of agent $1$ is $2$. And if only agent $3$ bids $c_3+1=5$, the best response is $5$. We conclude $\high{b}_1=5$ by taking the maximum among all the above best responses. We can also perform similar calculations for agent $2$, obtaining $\high{b}_2=6$ and $\low{b}_2=2$. For agent $3$, it always yields a utility of $0$ when other agents bid their cost or the cost plus one. Thus, $\high{b}_3=4$ by definition, and $\low{b}_3=4$.
\end{example}

The quantities $\high{b}_i$ and $\low{b}_i$ can be used to establish lower and upper bounds on the unit price of mixed NE of PC and PB. Theorem~\ref{t:3} establishes such a lower bound for PC. 

\begin{theorem}\label{t:3}
For any energy market $\mathcal{I}$, Pas-as-Clear admits at least one pure NE with unit price at least $\max_{i\preceq_\mathbf{c} \tau(\mathbf{c})} \high{b}_i$. Moreover, the unit price of any mixed NE of Pay-as-Clear is at least $\max_{i\preceq_\mathbf{c}\tau(\mathbf{c})}\low{b}_i$.
\end{theorem}


Theorem~\ref{t:PB} establishes a unit price upper bound at the worst-case mixed Nash equilibrium for PB.
\begin{theorem}\label{t:PB}
    Any mixed NE of Pay-as-Bid is supported in $[\max_{i\preceq_\mathbf{c}\tau(\mathbf{c})}\low{b}_{i}-1, \max_{i\preceq_\mathbf{c}\tau(\mathbf{c})}\high{b}_{i}]$.
\end{theorem}

It should now be clear how Theorem~\ref{t:2} is directly implied by Theorem~\ref{t:3} and Theorem~\ref{t:PB}. Specifically, Theorem~\ref{t:PB} establishes that the support of any mixed NE of PB lies between the unit price of the best and worst mixed NE of PC. In Theorem~\ref{thm:UB-PB-better}, we establish that under some mild assumption excluding degenerate instances, the worst-case NE of PB is strictly better than the worst-case NE of PC.

\begin{restatable}{theorem}{TUBPBbetter}
    \label{thm:UB-PB-better}
    The unit price given by the worst-case mixed NE of Pay-as-Bid is strictly smaller than $\max_{i\preceq_\mathbf{c}\tau(\mathbf{c})}\high{b}_i$ in a large family of instances.
\end{restatable}




The rest of the paper is structured as follows. In Sections~\ref{s:4}, \ref{s:PC_stratis}, and ~\ref{s:PB}, we provide informal proof sketches for Theorems~\ref{t:1},~\ref{t:3}, and~\ref{t:PB}, respectively. In Section~\ref{s:exp}, we present our experimental evaluation. The full proofs of all statements, as well as additional experiments appear in Appendix.

\section{Sketch of Proof of Theorem~\ref{t:1}}\label{s:4}
In this section, we present the main ideas of the proof of Theorem~\ref{t:1}, establishing that PB, PC, and VCG cannot be strictly dominated by any other mechanism $\mathcal{M}$. To do so, we construct an instance $\mathcal{I}$ at which any mechanism $\mathcal{M}$ admits at least one mixed NE with unit price strictly greater than the unit price of \textit{best mixed NE} of PB, PC, and VCG. The full proof of Theorem~\ref{t:1} is deferred to Appendix~\ref{proof:t:1}.
\smallskip
\smallskip

\noindent First consider the instance $\mathcal{I}$ with $n  = 2$ producers $A$ and $B$ with marginal costs $c_A = c_B = 0$ and supplies $s_A = s_B = 1$.  Notice that each of the producers can individually cover all the energy demand $Q = 1$. The bidding profile $\mathbf{b} = (b_A,b_B) =(0,0)$ is a pure Nash equilibrium for all PB, PC, and VCG. Thus, they all admit a NE with $0$ unit price. We will establish that any mechanism $\mathcal{M}$, regardless of $\vecx^\mathcal{M}(\vecb)$ and $\vecp^\mathcal{M}(\vecb)$, admits at least one mixed NE with positive unit price.
\smallskip
\smallskip

\noindent Let a sequence of positive numbers $\epsilon_1,\epsilon_2, \ldots, \epsilon_k$ with $\lim_{k \rightarrow \infty} \epsilon_k = 0$. For each $\epsilon_k >0$ consider the instance $\mathcal{I}_k$ where the marginal costs are $c_A = c_B = \epsilon_k >0$ and consider a mixed NE $\sigma^k = (\sigma^k_A , \sigma^k_B)$ of mechanism $\mathcal{M}$ in $\mathcal{I}_k$. Let the sequence $\mathcal{S}:= \sigma_1,\sigma_2,\ldots,\sigma_k,\ldots$ and $\sigma^\star$ of mixed NE and let  $\lim_{k \rightarrow \infty} \sigma^k = \sigma^\star$ be its limiting point\footnote{We can assume without loss of generality the existence of a limiting point $\sigma^\star = (\sigma^\star_A,\sigma^\star_B)$ since the product of simplices is a compact space and thus there is always a convergent subsequence, see also the full proof in Appendix~\ref{proof:t:1}.}. Since $\lim_{k \rightarrow \infty} \epsilon^k = 0$, the bidding profile $\sigma^\star$ is a mixed Nash equilibrium of mechanism $\mathcal{M}$ on the original instance $\mathcal{I}$ where $c_A = c_B = 0$.
\smallskip
\smallskip

\noindent Now consider the following two mutually exclusive cases for the mixed NE $\sigma^\star = (\sigma^\star_A, \sigma^\star_B)$:
\begin{enumerate}
\item $\sigma^\star_A(b_A) > 0$ and $\sigma^\star_B(b_B) > 0$ for some bids $b_A , b_B \geq 1$.

\smallskip
\smallskip

\item $\sigma^\star_A(0) = 1$ or $\sigma^\star_B(0) = 1$.
\end{enumerate}

\noindent For Case~$1$, it is easy to see that $\sigma^\star = (\sigma^\star_A,\sigma^\star_B)$ admits a positive unit price. More precisely, the event that both agents bid a value greater or equal to $1$ is at least $\sigma^\star_A(b_A) \cdot \sigma^\star_B(b_B) > 0$. Then, by individual rationality, the mechanism pays at least $1$ in such an event. Meaning the expected unit price of $\mathcal{M}$ is positive. 
\smallskip
\smallskip

\noindent Case~$2$ is significantly more challenging. Without loss of generality, we assume $\sigma_A^\star(0) =1$. Let us go back to the sequence $\sigma_1,\sigma_2,\ldots,\sigma^k \rightarrow \sigma^\star$. Recall that $\sigma^k = (\sigma^k_A,\sigma^k_B)$ is a mixed NE of the perturbed game with marginal costs $s_A = s_B = \epsilon_k$. Let producer $A$ sell with non-zero probability a positive amount of energy\footnote{The case where agent $A$ does not sell any energy is presented in the full proof, see Appendix~\ref{proof:t:1}.}. Since $\sigma^k$ is a mixed NE, and agent $A$ admits $c_A = \epsilon_k$ marginal cost, the expected payment of $\mathcal{M}$ to agent $A$ must be positive (otherwise agent $A$ admits negative expected revenue). Since  $\sigma^k_A(0) > 0$\footnote{The latter can be assumed without loss of generality since $\lim_{k \rightarrow \infty} \sigma^A(0) = 1$.} and $\sigma^k$ is a NE, we are ensured that 
the expected payment of $A$ when bidding $0$ must be positive,
\begin{equation}\label{eq:1}
\sum_{b_B \in [M]} \sigma_B^k(b_B) \cdot x_A(0,b_B)\cdot p_A(0,b_B) >0.    
\end{equation}
Let $C_k := \{b_B \in [M]~:~\sigma^k_B(b_B) > 0 \text{ and } x_A(0,b_B)\cdot p_A(0,b_B) > 0\}$ be the set of bids of producer $B$ that lead to a positive payment of producer $A$ once $b_A = 0$. Thus, Equation~\ref{eq:1} establishes that $C_k$ is not empty. Now consider the mixed strategy $\hat{\sigma}^k_B $ defined as,
\[\hat{\sigma}^k_B(j)  := \begin{cases}
      \frac{\sigma^k_B(j)}{1 -  \sum_{\ell \notin C_k}\sigma^k_B(\ell)}, &\text{if }j \in C_k, \\
      0, &\text{otherwise.}\\
\end{cases}
\]
The cornerstone idea our proof is that the bidding profile $\hat{\sigma}^k:= (0, \hat{\sigma}^k_B)$
satisfies the following two properties $i)$ $\hat{\sigma}^k$ is an $(1-M\cdot\sigma_A^k(0))$-approximate mixed NE $ii)$
the expected payment of producer $A$ is at least $\mu_{\mathcal{M}} := \min_{b_B \in C_k} x_A(0,b_B)\cdot p_A(0,b_B)$. Both of the properties are formally established in Lemma~\ref{l:main} in Appendix~\ref{proof:t:1}. 

To this end, we are ready to provide the final step of the proof. Given the sequence $\sigma_1,\sigma_2,\ldots,\sigma_k$ we construct the new sequence $\hat{\sigma}_1,\hat{\sigma}_2,\ldots,\hat{\sigma}_k$. We know that each $\hat{\sigma}_k$  is an $(1-\sigma_A^k(0))$-approximate NE and its expected payment is at least $\mu_\mathcal{M} > 0$. Since $\lim_{k \rightarrow \infty} \sigma^k_A(0) = 1$ we are ensured that $\hat{\sigma} = \lim_{k \rightarrow \infty} \hat{\sigma}_k$ is a mixed NE of the original game while its expected payment to producer $A$ is at least $\mu_{\mathcal{M}} > 0$.

\section{Analyzing Pay-as-Clear}\label{s:PC_stratis}

In this section, we provide the high-level ideas behind Theorem~\ref{t:3}. In order to build intuition, we start with a simple observation that PC is not a \textit{truthful} mechanism. 

\begin{corollary}
\label{cor:PC-not-truthful}
Pay-as-Clear is not a truthful mechanism.
\end{corollary}
\begin{proof}
Consider an energy market with $n=2$ producers with  supplies $s_1=s_2=1$ and marginal costs $c_1=0,c_2=M$. If producer $2$ bids $M$, then agent $1$ has positive payoff by bidding in $\{1,\ldots,M-1\}$. As a result, bidding $0$ is not a weakly dominating strategy for agent $1$.
\qed   
\end{proof}
\begin{remark}
\label{remark:PC-not-truthful}
We remark that in almost every energy market $\mathcal{I}$,  PC is not truthful. In proof of Corollary~\ref{cor:PC-not-truthful}, the key idea is that the pivotal agent can increase the clearing price without changing the allocation. In most cases, the pivotal agent in the truthful bidding profile $\mathbf{c}$ has this ability. Intuitively, the only condition is that the bid of the next agent of the pivotal agent is not equal to the clearing price. In formal terms, if $c_{\tau(\mathbf{c})}=q(\mathbf{c})<c_{\tau(\mathbf{c})+1}$, the pivotal agent in the truthful bidding profile $\mathbf{c}$ can improve her utility by increasing the clearing price to $c_{\tau(\mathbf{c})+1}$ without changing the allocation.
\end{remark}
\noindent Up next, we provide the idea behind the proof of Theorem~\ref{thm:UB-PC} the proof of which can be found in Appendix~\ref{proof:thm:UB-PC}. Theorem~\ref{thm:UB-PC} not only provides a lower bound on the worst-case NE of PC but also establishes the existence of pure NE.
\begin{restatable}{theorem}{TUBPC}
    \label{thm:UB-PC}
Pay-as-Clear always admits a pure NE while the unit price of the worst-case pure NE is at least $\max_{i\preceq_\mathbf{c} \tau(\mathbf{c})} \high{b}_i$.
\end{restatable}

\begin{proof}[Proof Sketch]
To simplify things let us consider $\high{b}_i := \argmax_{b_i \in [M]} U_i(b_i, \mathbf{c}_{-i})$, i.e. the best response of producer $i \in [n]$ once all other agents bid truthfully. Let $i^\star$ denote the agent $\preceq_\mathbf{c}\mathbf{c}$ with maximum $\high{b}_i$, $i^\star := \argmax_{i \preceq_{\mathbf{c}} \tau(\mathbf{c})} \high{b}_i$. In the case there are multiple such agents, we choose the one with the highest index. We will show that the profile $\textbf{b}^\star:=(\high{b}_{i^\star}, \mathbf{c}_{-i^\star})$ provide a unit price of $\high{b}_i$, and $\vecb^\star$ is a pure NE.
\smallskip

\noindent The first step is to argue that the unit price of $\textbf{b}^\star$ is at least $\high{b}_{i^\star} = \max_{i\preceq_\mathbf{c} \tau(\mathbf{c})} \high{b}_i$. Since $\high{b}_{i^\star}$ is the best response to $\mathbf{c}_{-i^\star}$, 
agent $i^\star \in [n]$ sells a positive amount of energy meaning that the price for all agents is at least $\high{b}_{i^\star}:= \max_{i\preceq_\mathbf{c} \tau(\mathbf{c})} \high{b}_i$. We remind that by the definition of $\tau(\mathbf{c})$ (Definition~\ref{def:allocation-and-clearing-price}), we know that all agents $i\preceq_\mathbf{c} \tau(\mathbf{c})$ collectively can cover the whole energy demand, $\sum_{i\preceq_\mathbf{c} \tau(\mathbf{c})} s_i \geq 1$. As a result, $i^\star \in [n]$ is the pivotal agent and thus the unit energy price is exactly $\high{b}_{i^\star}$.
\smallskip
\smallskip

\noindent Up next we establish that $\mathbf{b}^\star$ is a pure NE. By definition, agent $i^\star$ has no incentive to deviate. We now show that any agent $j\preceq_\mathbf{c} \tau(\mathbf{c})$ has no incentive to deviate. Notice that in the bidding profile $(\high{b}_{i^\star},\mathbf{c}_{-i^\star})$, agent $j$ sells all of her energy $s_j$. This is because $q(\vecb^\star)>\high{b}_{i^\star} \geq \max_{j\preceq_\mathbf{c} \tau(\mathbf{c})} \high{b}_j \geq \max_{j\preceq_\mathbf{c}\tau(\mathbf{c})} c_j\geq c_j$ and the agent $i^\star \in [M]$ sells a positive amount of energy with bid $\high{b}_{i^\star}$.  And for any bid $b_j \leq \high{b}_{i^\star}$, $U_j(b_j,\textbf{b}^\star_{-j})  = U_j(\textbf{b}^\star)=(\high{b}_{i^\star} - c_j) \cdot s_j$ since the agent $j$ will still sell $s_j$ amount of energy at price $\high{b}_{i^\star}$.
\smallskip
\smallskip

\noindent Let us now see the utility of an agent $j\preceq_\mathbf{c} \tau(\mathbf{c})$ once deviating to a bid $b_j> \high{b}_{i^\star}$. In this case, the utility of agent $j \in [n]$ would be exactly the same as in the bidding profile $(b_j,\mathbf{c}_{-j})$. As a result, $U_j(b_j, \mathbf{b}_{-j}^\star) = U_j(b_j, \mathbf{c}_{-j}) \leq U_j(\high{b}_j, \mathbf{c}_{-j})$ where the last inequality follows by the definition of $\high{b}_j$.
Notice that $U_j(\high{b}_j, \mathbf{c}_{-j}) \leq  s_j \cdot (\high{b}_j - c_j) \leq  s_j \cdot (\high{b}_{i^\star} - c_j)  = U_j(\textbf{b}^\star)$. As a result, we overall get that for any bid $b_j > \high{b}_{i^\star}$, $U_j(b_j,\textbf{b}^\star_{-j})  \leq U_j(\textbf{b}^\star)$. Thus, $U_j(b_j,\textbf{b}^\star_{-j})  \leq U_j(\textbf{b}^\star)$ for any bid $b_j \in [M]$.
\smallskip
\smallskip

\noindent Any agent $j\succ_\mathbf{c} \tau(\mathbf{c})$ with $c_j \geq \high{b}_{i^\star}$ sells $0$ energy at strategy profile $\mathbf{b}^\star$ since $\sum_{j\preceq_\mathbf{c} \tau(\mathbf{c})} s_j \geq 1$. The same holds for any bids $b_j \geq c_j$ while bids $b_j < c_j$ can lead even to negative utility. Finally, any agent $j\succ_\mathbf{c} \tau(\mathbf{c})$ with $c_j <  \high{b}_{i^\star}$ admits utility $U_j(\mathbf{b}^\star) = s_j\cdot (\high{b}_{i^\star} - c_j) > 0$ sells $s_j$ amount of energy at price $\high{b}_{i^\star}$. Agent $j$ will get the exact same utility for any bid $b_j  \leq \high{b}_{i^\star}$. However for any bid $b_j > \high{b}_{i^\star}$ agent $j$ sells $0$ energy since $\sum_{j\preceq_\mathbf{c} \tau(\mathbf{c})} s_j \geq 1$.

\qed
\end{proof}

\noindent We also provide a lower bound of the unit price given by the best-case equilibrium. We notice that every agent can always achieve the utility given by the best response to the truthfully bidding profile, regardless of others' behavior. Lemma \ref{lem:least-utility-PC} presents the formal statement. And the proof of Lemma \ref{lem:least-utility-PC} is postponed to Appendix \ref{proof:lem:least-utility-PC}. Using this property, Theorem \ref{thm:LB-PC} presents the lower bound of the unit price given by any equilibrium. The formal proof is postponed to Appendix \ref{proof:thm:LB-PC}.

\begin{restatable}{lemma}{LleasttuilityPC}
    \label{lem:least-utility-PC}
    In Pay-as-Clear, consider a bidding profile $\vecb$, we have $U_i(\mathbf{b})\geq U_i(b_i,\mathbf{c}_{-i})$. As the result, in any mixed NE $\sigma$, we have $\Ex_{\mathbf{b}\sim\sigma}[U_i(\mathbf{b})]\geq \max_{b'_i\in[M]}U(b'_i,\mathbf{c}_{-i}).$
\end{restatable}

\begin{restatable}{theorem}{TLBPC}
    \label{thm:LB-PC}
    The unit price given by the best-case mixed NE of Pay-as-Clear is at least $\max_{i\preceq_\mathbf{c}\tau(\mathbf{c})}\low{b}_i$.
\end{restatable}

\section{Analyzing Pay-as-Bid}
\label{s:PB}
In this section, we provide the high-level ideas behind Theorem \ref{thm:LB-PB} and Theorem \ref{thm:UB-PB}, providing respectively a lower bound and an upper bound to the unit price given by the mixed NEs of PB. Additionally, we introduce a better upper bound by introducing a mild assumption in Theorem \ref{thm:UB-PB-better}. 

\noindent We start with a simple example showing that pure NE may not exist in PB. 

\begin{corollary}\label{c:15}
Pay-as-Bid may not have a Pure Nash equilibrium.
\end{corollary}

\begin{proof}
Consider an energy market with $n=2$ producers with  supplies $s_1=s_2=3/4$ and marginal costs $c_1=c_2=0$ with $M\geq 5$. 
For any bidding profile where agent $2$ deterministically bids $b$, the agent $1$'s best response is always either bidding $b$ or bidding $M$. In the first case, agent $1$ will sell $3/4$ energy at price $b$ leading to $3b/4$ utility, while the second agent $1$ will sell $1/4$ energy at price $M$ leading to $M/4$ utility. If we are in the first case, it follows $3b/4 \geq M/4$. When agent $2$ bids $b-1$, her utility is improved by $3(b-1)/4-b/4>(2M/3-3)/4>0$. If we are in the second case, where $3b/4< M/4$.  Then bidding $M$ increases the payoff of agent $2$. Thus, there is no pure Nash equilibrium.\qed
\end{proof}

The proof of Corollary~\ref{c:15} reveals a crucial phenomenon in the \textit{best-response dynamics} of PB: The agents repeatedly \textit{marginally decrease} their bids to sell their whole supply (see Fig~\ref{fig:sub1}). At some point, the clearing price is so low that
agent $i \in [n]$ prefers to sell a \textit{small fraction of her supply} at a very high price (see Fig~\ref{fig:sub2}). 
\begin{figure}[htbp]
  \centering
  \begin{minipage}[t]{0.32\textwidth}
    \centering
    \begin{tikzpicture}
      \begin{axis}[axis x line=middle, axis y line=none,width=1.2\linewidth,
        xmin=0, xmax=105, ymin=-1, ymax=2,
        xtick={0,40,90,100}, 
        xticklabels={$0$,$\low{b}_i$,$\high{b}_i$,$M$},
        tick style={thick},
        clip=false, axis line style={->, thick}]
        \draw[->, thick] (axis cs:100,0.15) -- (axis cs:100,0.01);
        \node[above] at (axis cs:100,0.14) {$i$};
        \draw[->, thick] (axis cs:95,0.15) -- (axis cs:95,0.01);
        \node[above] at (axis cs:95,0.12) {$j$};
        \draw[->, thick, dashed,dash pattern=on 2pt off 1pt] (axis cs:100,0) to [out=100,in=80] (axis cs:90,0);
        \draw[->, thick, dashed,dash pattern=on 2pt off 1pt] (axis cs:90,0) to [out=100,in=80] (axis cs:80,0);
        \draw[->, thick, dashed,dash pattern=on 2pt off 1pt] (axis cs:80,0) to [out=100,in=80] (axis cs:70,0);
        \draw[->, thick, dashed,dash pattern=on 2pt off 1pt] (axis cs:95,0) to [out=-100,in=-80] (axis cs:85,0);
        \draw[->, thick, dashed,dash pattern=on 2pt off 1pt] (axis cs:85,0) to [out=-100,in=-80] (axis cs:75,0);
        \draw[->, thick, dashed,dash pattern=on 2pt off 1pt] (axis cs:75,0) to [out=-100,in=-80] (axis cs:65,0);
        \draw[->, white, thick, dashed] (axis cs:40,0.01) to [out=60,in=120] (axis cs:90,0.01);
        \draw[->, thick, dashed,white,dash pattern=on 0.1pt off 10pt] (axis cs:35,-0.01) to [out=-60,in=-120] (axis cs:85,-0.01);
      \end{axis}
    \end{tikzpicture}
    \caption{}
    \label{fig:sub1}
  \end{minipage}
  \hfill
  \begin{minipage}[t]{0.32\textwidth}
    \centering
    \begin{tikzpicture}
      \begin{axis}[axis x line=middle, axis y line=none,width=1.2\linewidth,
        xmin=0, xmax=105, ymin=-1, ymax=2,
        xtick={0,40,90,100}, 
        xticklabels={$0$,$\low{b}_i$,$\high{b}_i$,$M$},
        tick style={thick},
        clip=false, axis line style={->, thick}]
        \draw[->, thick] (axis cs:40,0.15) -- (axis cs:40,0.01);
        \node[above] at (axis cs:40,0.14) {$i$};
        \draw[->, thick] (axis cs:35,0.15) -- (axis cs:35,0.01);
        \node[above] at (axis cs:35,0.12) {$j$};
        \draw[->, thick, dashed,dash pattern=on 2pt off 1pt] (axis cs:40,0) to [out=60,in=120] (axis cs:90,0);
        \draw[->, thick, dashed,dash pattern=on 2pt off 1pt] (axis cs:45,0) to [out=-100,in=-80] (axis cs:35,0);
        \draw[->, thick, dashed,white,dash pattern=on 0.1pt off 10pt] (axis cs:35,-0.01) to [out=-60,in=-120] (axis cs:85,-0.01);
      \end{axis}
    \end{tikzpicture}
    \caption{}
    \label{fig:sub2}
  \end{minipage}
   \hfill
    \begin{minipage}[t]{0.32\textwidth}
    \centering
    \begin{tikzpicture}
      \begin{axis}[axis x line=middle, axis y line=none, width=1.2\linewidth,
        xmin=0, xmax=105, ymin=-1, ymax=2,
        xtick={0,40,90,100}, 
        xticklabels={$0$,$\low{b}_i$,$\high{b}_i$,$M$},
        tick style={thick},
        clip=false, axis line style={->, thick}]
        \draw[->, thick] (axis cs:90,0.15) -- (axis cs:90,0.01);
        \node[above] at (axis cs:90,0.14) {$i$};
        \draw[->, thick] (axis cs:35,0.15) -- (axis cs:35,0.01);
        \node[above] at (axis cs:35,0.12) {$j$};
        \draw[->, thick, dashed,dash pattern=on 2pt off 1pt] (axis cs:90,0) to [out=100,in=80] (axis cs:80,0);
        \draw[->, thick, dashed,dash pattern=on 2pt off 1pt] (axis cs:80,0) to [out=100,in=80] (axis cs:70,0);
        \draw[->, thick, dashed,dash pattern=on 2pt off 1pt] (axis cs:85,0) to [out=-100,in=-80] (axis cs:75,0);
        \draw[->, thick, dashed,dash pattern=on 2pt off 1pt] (axis cs:75,0) to [out=-100,in=-80] (axis cs:65,0);
        \draw[->, thick, dashed,dash pattern=on 2pt off 1pt] (axis cs:35,-0.01) to [out=-60,in=-120] (axis cs:85,0);
        \draw[->, white, thick, dashed] (axis cs:40,0.01) to [out=60,in=120] (axis cs:90,0.01);
      \end{axis}
    \end{tikzpicture}
    \caption{}
    \label{fig:sub3}
  \end{minipage}
\end{figure}
'\smallskip

\noindent The cornerstone idea of our analysis is that the price threshold for agent $i \in [n]$ is exactly $\low{b}_i:= \lceil c_i + \max_{b'_i\in[M]}U_i(b'_i,\mathbf{c}_{-i}) / s_i\rceil$. This is because even if agent $i \in [n]$ sells her whole energy $s_i \in [0,1]$ at price strictly lower than $\low{b}_i$ then her payoff smaller than $\max_{b'_i\in[M]}U_i(b'_i,\mathbf{c}_{-i})$. However agent $i \in [n]$ can always guarantee utility $\max_{b'_i\in[M]}U_i(b'_i,\mathbf{c}_{-i})$ by bidding $\high{b}_i$\footnote{If each agent $j \neq i$ bids greater or equal her marginal cost $b_j \geq c_j$ then agent $i \in [n]$ gets at least $\max_{b'_i\in[M]}U_i(b'_i,\mathbf{c}_{-i})$ by bidding $\high{b}_i$.}. 

        



Once agent $i \in [n]$ has moved to $\high{b}_i$ then the rest agent $j \neq i$ can now sell their whole supply $s_j$ at a price almost $\high{b}_i$ (see Fig.~\ref{fig:sub3}). This will lead to a new decreasing phase since agent $i \in [n]$ now would be willing to decrease her bid to sell her whole supply $s_i \in [0,1]$. This will lead to a new sequence of price-decreases (see Fig.~\ref{fig:sub1}) until agent $i \in [n]$ moves to $\high{b}_i$ (see Fig.~\ref{fig:sub2}) and the whole process will cyclically repeat. 
\begin{remark}
The first agent $i \in [n]$ that will no longer be willing to further decrease the price (Fig.~\ref{fig:sub1}) but rather switch to $\high{b}_i$ is the agent with $i^\star:= \mathrm{argmax}_{i \in [n]} \low{b}_i$. Intuitively, a best-response sequence of PB will admit this cyclic trajectory on the interval $[\low{b}_{i^\star},\high{b}_{i^\star}]$ that is smaller than $\max_{i}\high{b}_i$ which is the lower bound of PC.   
\end{remark}

Building on the above intuition on the best-response dynamics in PB, we establish that any mixed NE of PB approximately lies on the interval $[\low{b}_{i^\star},\high{b}_{i^\star}]$. The first step towards this direction is establishing that any agent $i \in [n]$ can always get utility $\max_{b'_i\in[M]} U(b'_i,\mathbf{c}_{-i})$ regardless of the other agents' bids. The latter is formally established in Lemma~\ref{lem:least-utility-PB}, the proof of which is deferred to Appendix~\ref{proof:lem:ub-low-cost-agent}.

\begin{restatable}{lemma}{LleasttuilityPB}
\label{lem:least-utility-PB}
In Pay-as-Bid, consider a bidding profile $\vecb$, we have $U_i(\mathbf{b})\geq U_i(b_i,
\vecb'_{-i})$ when $b'_j\leq b_j$ for any $j\neq i$. As a result, in any mixed NE $\sigma$, we have that $\Ex_{\mathbf{b}\sim\sigma}[U_i(\mathbf{b})]\geq \max_{b'_i\in[M]}U(b'_i,\mathbf{c}_{-i}).$
\end{restatable}

Lemma~\ref{lem:least-utility-PB} established that any bid lower than $\low{b}_i$ will not appear in an mixed NE since it gives a lower utility than $\max_{b'_i\in[M]} U(b'_i,\mathbf{c}_{-i})$ even if agent $i$ sells all of her supply $s_i \in [0,1]$. We can then further establish that an agent $i \in [n]$ will not bid lower than $\max_{i}\low{b}_{i^\star}-1$ since agent $i^\star := \argmax_{i}\low{b}_{i}$ always bids at least $\low{b}_{i^\star}$. The latter is formally established in Theorem~\ref{thm:LB-PB}, providing a lower bound on the support of any mixed NE of PC (see Appendix~\ref{proof:thm:LB-PB} for the proof).
\begin{restatable}{theorem}{TLBPB}
    \label{thm:LB-PB}
    In any mixed NE $\sigma$, every agent bids at least $\max_{i\preceq_{\mathbf{c}}\tau(\mathbf{c})} \low{b}_i -1$ with probability of $1$. Thus, the unit price given by $\sigma$ is at least $\max_{i\preceq_{\mathbf{c}}\tau(\mathbf{c})} \low{b}_i-1$.
\end{restatable}

Theorem~\ref{thm:LB-PB} establishing that at any mixed NE of Pay-as-Bid all agents bid at least $\max_{i\preceq_{\mathbf{c}}\tau(\mathbf{c})} \low{b}_i-1$. As a result, in case agent $i \in [n]$ bids $\low{b}_i$ then it is guaranteed to sell all of her energy $s_i \in [0,1]$ at price $\low{b}_i$ resulting in utility $U_i(\high{b}_i,\mathbf{c}_{-i})$. The next crucial step of our analysis is establishing that once an agent bids higher than $\high{b}_i$, her utility is strictly smaller than $U_i(\high{b}_i,\mathbf{c}_{-i})$. By leveraging the latter, we are able to establish that any bid $b_i > \high{b}_i$ must necessarily admit $0$ probability in a mixed NE. The latter is formally established in Theorem~\ref{thm:UB-PB}, the proof of which is deferred in Appendix~\ref{proof:thm:UB-PB}.

\begin{restatable}{theorem}{TUBPB}
    \label{thm:UB-PB}
The support of any mixed NE $\sigma$ of Pay-as-Bid is at most $\max_{i\preceq_\mathbf{c}\tau(\mathbf{c})}\high{b}_i$.
\end{restatable}

Theorem~\ref{t:2} follows directly by  Theorem~\ref{thm:LB-PB} and Theorem~\ref{thm:UB-PB}. With a more technically complicated analysis presented in Appendix~\ref{proof:thm:UB-PB-better}, we can establish Theorem~\ref{thm:UB-PB-better} that provides an even tighter bound. 

\section{Experimental Evaluation}

\label{s:exp}
\begin{figure}[htbp]
  \centering
  \begin{minipage}[t]{0.48\textwidth}
    \centering
    \includegraphics[height=4cm]{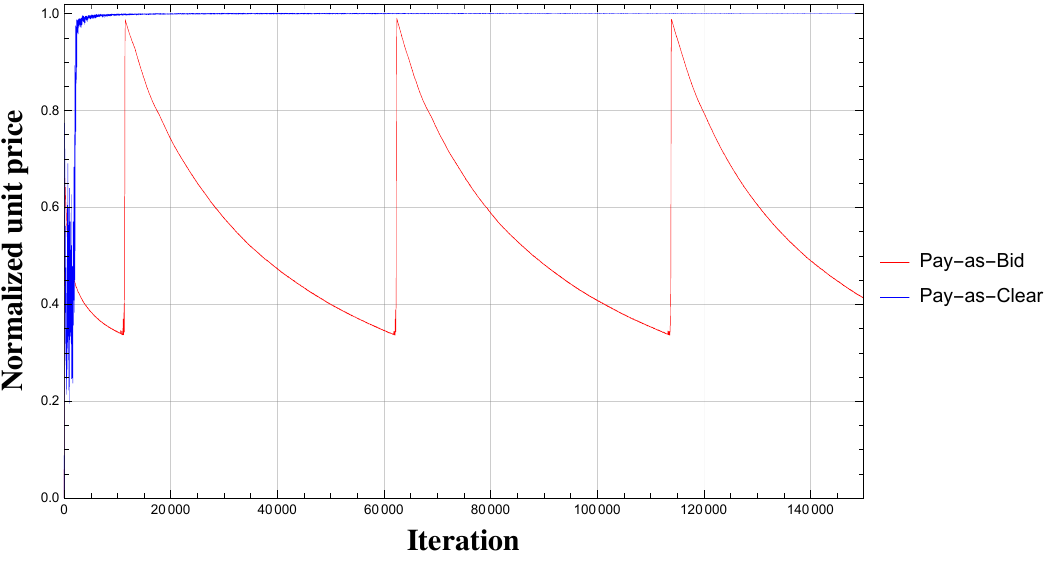}
    \caption{The normalized unit price changing by iteration in an energy market with $M=800$ and $n = 4$ producers with supplies $s_1=s_2=s_3=s_4=0.3$ and marginal costs $c_1=c_2=c_3=c_4=0$. For every agent $i \in [n]$, we have $\high{b}_i=800$ and $\low{b}_i=267$.}
    \label{fig:s4agent}
  \end{minipage}
  \hfill
  \begin{minipage}[t]{0.48\textwidth}
    \centering
    \includegraphics[height=4cm]{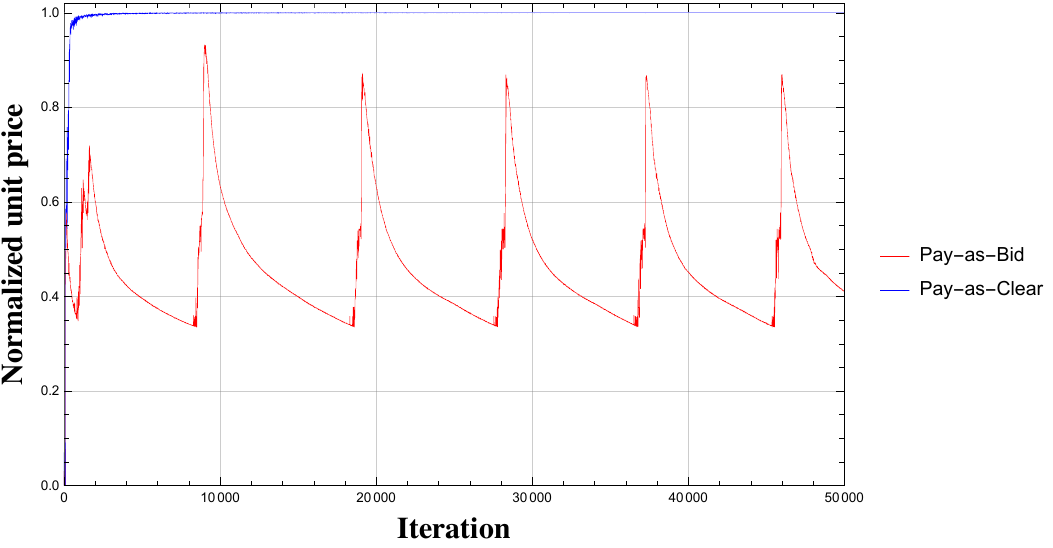}
    \caption{The normalized unit price changing by iteration in an energy market with $M=900$ and $n=3$ producers with supplies $s_1=0.5, s_2=0.75, s_3=0.25$ and marginal costs $c_1=c_2=0, c_3=300$.}
    \label{fig:PC-may-be-better}
  \end{minipage}
\end{figure}

\noindent We simulate the dynamic evolution of unit prices in an energy market where all producers utilize a no-regret learning algorithm for bidding (see also Appendix~\ref{a:exp} for additional experimental evaluations).

We employ the Hedge algorithm, initialized with a uniform distribution over the strategy space~\cite{H17}. In each iteration, the Hedge algorithm updates its beliefs based on the bidding profile sampled from the joint distribution of the agents' mixed strategies.

Our simulations are conducted on several representative market instances, including the illustrative example presented previously. Figure \ref{fig:s4agent} displays the evolution of the normalized unit price for both PC and PB in a symmetric market with four agents. In the PC mechanism, bids rapidly converge to an equilibrium where the unit price aligns with the upper bound established by Theorem \ref{thm:UB-PC}. In contrast, the unit price in the PB mechanism exhibits a distinct periodic behavior. Initially, the unit price quickly approaches $\max_{i\preceq_{\mathbf{c}}\tau(\mathbf{c})}\high{b}_i$, consistent with the upper bound provided by Theorem \ref{thm:UB-PB}. Subsequently, the price gradually declines towards $\max_{i\preceq_{\mathbf{c}}\tau(\mathbf{c})}\low{b}_i-1$, which corresponds to the lower bound specified by Theorem \ref{thm:LB-PB}. Following this decrease, the unit price sharply jumps back to approximately $\max_{i\preceq_{\mathbf{c}}\tau(\mathbf{c})}\high{b}_i$, and this cyclical pattern repeats. Figure \ref{fig:PC-may-be-better} further illustrates the evolution of the unit price in the market described by Example \ref{example:best-PC-can-be-better} with $\delta=300$. We observe a similar phenomenon and a lower unit price in PB, despite the existence of a pure NE in PC that outperforms any mixed NE in PB in this instance.\footnote{The complete construction and justification are provided in Appendix~\ref{app:bestPC}.} Reaching this superior pure NE in PC is challenging, as the convergence of the Hedge algorithm is highly sensitive to initialization and randomness, due to the simplistic structure of PC. Crucially, in these market instances, our experiments consistently demonstrate that the time-average unit price under PB is significantly lower than that under PC, particularly when agents deploy no-regret learning algorithms. More plots are left in Appendix~\ref{a:exp} for verification.

\section{Conclusion} The central result of this work is that the worst-case Nash equilibrium under the Pay-as-Bid mechanism always yields lower energy prices than the worst-case NE under the Pay-as-Clear mechanism. While this does not preclude the existence of specific mixed NEs under PC with lower prices than certain mixed NEs under PB, it nonetheless provides the strongest possible worst-case guarantee. Specifically, we show that PC, like PB and the VCG mechanism, cannot be strongly dominated by any alternative mechanism.

Our experimental evaluations support our theoretical findings, indicating that online learning dynamics consistently lead to lower average unit prices under PB compared to PC. An intriguing direction for future research is to derive formal price guarantees for both PB and PC under the assumption that all producers employ no-regret online learning algorithms to determine their bids.

\subsection*{Acknowledgements}
The work of Zhile Jiang and Stratis Skoulakis was funded by the Villum Young Investigator Award no.~72091.

\bibliographystyle{splncs04}
\bibliography{reference}

\clearpage

\appendix

\section{Proof of Theorem~\ref{t:1}}\label{proof:t:1}

\Tnodominating*

\begin{proof}

Consider the agents $A$ and $B$, both of which admit marginal cost $0$ and supply $1$. Up next, we show that any mechanism $\mathcal{M}$ must admit at least one mixed NE with a positive unit price.  
\smallskip

\noindent For each $k \in \mathbb{N}$, assume that the agents $A$ and $B$ admit marginal cost $\epsilon_k = 2^{-k}$ (instead of $0$) and consider a mixed NE $\sigma^k = (\sigma^k_A,\sigma^k_B)$ for this different two-player game. Let the sequence  $\mathcal{S}:=\sigma^1,\sigma^2,\ldots,\sigma^k,\ldots$ for each $k \in \mathrm{N}$ and notice that there are the following mutually exclusive cases.   
\begin{enumerate}
\item There exists an infinite subsequence $\mathcal{S}' \subseteq 
\mathcal{S}$ such that for any $\sigma_k \in \mathcal{S}'$,  $\mathrm{support}(\sigma^k_A) \neq \{0\}$ and $\mathrm{support}(\sigma^k_B) \neq \{0\}$

\smallskip
\smallskip

\item There exists an infinite subsequence $\mathcal{S}' \subseteq 
\mathcal{S}$ such that for any $\sigma_k \in \mathcal{S}'$,  $\mathrm{support}(\sigma^k_A) = \{0\}$
\end{enumerate}

\noindent Our proof consists of showing in any case the
sequence $\mathcal{S}'$ converges to a point $\lim_{k \rightarrow \infty} \sigma^k = \sigma^\star$ that is a mixed NE for the original game (marginal costs of the agent are $0$) while the payment of the mechanism at $\sigma^\star$ is positive. We will establish the latter for two cases separately.

\begin{itemize}
\item \ul{There exists an infinite subsequence $\mathcal{S}' \subseteq 
\mathcal{S}$ such that for any $\sigma_k \in \mathcal{S}'$,  $\mathrm{support}(\sigma^k_A) \neq \{0\}$ and $\mathrm{support}(\sigma^k_B) \neq \{0\}$}.
\smallskip

Let $\sigma^\star = (\sigma^\star_A, \sigma^\star_B)$ be the limiting point of $\mathcal{S}'$. If there exist $i \geq 1$ and $j \geq 1$ such that $\sigma^\star_A(i) > 0$ and $\sigma^\star_B(j) > 0$ then the expected payment of $\sigma^\star$ is at least \[\sigma^\star_A(i) \cdot \sigma^\star_B(j) \cdot \left(x_A(i,j)\cdot p_A(i,j) + x_B(i,j)\cdot p_B(i,j)\right)  > 0\]
where the first inequality follows by individual rationality. Since $\lim_{k \rightarrow \infty} \epsilon_k = 0$ we know that $\sigma^\star$ is mixed NE for the original game (marginal costs $0$) and admits positive payment.
\smallskip
\smallskip

\noindent As a result, we just need to consider the case where $\sigma^\star_A =e_0 = (1,0,\ldots,0)$ (the case $\sigma^\star_B =e_0$ follows symmetrically). To complete the proof we will use the Lemma~\ref{l:main}. 

\begin{lemma}\label{l:main}
Let the sequence $\mathcal{S}' = \sigma^1,\sigma^2, \ldots$ then for any mixed NE $\sigma^k = (\sigma^k_A,\sigma^k_B)$ we can construct a strategy pair $\hat{\sigma}^k = (\hat{\sigma}^k_A,\hat{\sigma}^k_B)$ such that $i)$ $\hat{\sigma}^k$ is an $(1-M\cdot\sigma^k_A(0))$-approximate mixed NE $ii)$ the payment of $\hat{\sigma}^k$ is greater than $\mu_{M} > 0$ where is a mechanism-dependent constant. 
\end{lemma}
Before presenting the proof of Lemma~\ref{l:main}, we complete the proof of Theorem~\ref{t:1}. Lemma~\ref{l:main} establishes that there exists a sequence $\hat{\sigma}^1,\hat{\sigma}^2,\ldots,\hat{\sigma}^k,\ldots$ where each $\hat{\sigma}^k = (\hat{\sigma}^k_B, \hat{\sigma}^k_B)$ a $(1-M\cdot\sigma_A^k(0))$-approximate mixed NE and admits payment at least $\mu_G > 0$. Since we are in compact space there exists a convergent subsequence, $\lim_{k \rightarrow \infty} \sigma_k = \sigma^\star$. Since $\lim_{k \rightarrow \infty} \epsilon_k = 0$ and $\lim_{k \rightarrow \infty} \sigma_A^k(0) = 1$, we are ensured that $\sigma^\star$ is a mixed NE while admits payment at least $\mu_G > 0$.

\smallskip
\smallskip

\paragraph{Proof of Lemma~\ref{l:main}.}
\noindent Let the mixed NE $\sigma^k = (\sigma^k_A,\sigma^k_B)$. We first consider the case where agent $A$ sells a positive amount of energy. Since $A$ has marginal cost $\epsilon_k > 0$, agent $A$ needs to receive a positive payment (otherwise its utility is negative while $A$ can make it non-negative by bidding $1$).
\smallskip
\smallskip

Since the overall payment of $A$ must be positive and $0 \in \mathrm{support}(\sigma_A^k)$, the expected payment of $A$ when bidding $0$ must be positive,
\[ \sum_{j \in [M]} \sigma_B^k(j) \cdot x_A(0,j)\cdot p_A(0,j) >0. \]
As a result, the set $C_k := \{j \in [M]~:~\sigma^k_B(j) > 0 \text{ and } x_A(0,j)\cdot p_A(0,j) > 0\}$ is not empty. Now consider the probability distribution $\hat{\sigma}^k_B $ defined as,
\[\hat{\sigma}^k_B(j)  := \begin{cases}
      \frac{\sigma^k_B(j)}{1 -  \sum_{\ell \notin C_k}\sigma^k_B(\ell)}, &\text{if }j \in C_k, \\
      0, &\text{otherwise.}\\
\end{cases}
\]
We are going to show that $\hat{\sigma}_k := (e_0,\hat{\sigma}_B^k)$ 
where $e_0 = (1,0,\ldots,0)$ is an $\epsilon_k$-approximate mixed NE while it admits a positive payment.
Since $0 \in \mathrm{support}(\sigma_A^k)$ we are ensured that 
\[ \sum_{j \in [M]} \sigma_B^k(j) \cdot x_A(0,j)\cdot (p_A(0,j) - \epsilon_k) \geq \sum_{j \in [M]} \sigma_B^k(j) \cdot x_A(i,j)\cdot (p_A(i,j) - \epsilon_k)~~~~\text{for all bids } i \in [M]\]
By the definition of $j \in C_k$ we get that 
\[\sum_{j \notin C_k} \sigma_B^k(j) \cdot x_A(0,j)\cdot (p_A(0,j) - \epsilon_k) \leq 0\]
At the same time, for all bids $i \geq 1$
\[\sum_{j \notin C_k} \sigma_B^k(j) \cdot x_A(i,j)\cdot (p_A(i,j) - \epsilon_k) \geq 0\]
due to the fact that $p_A(i,j)  \geq 1$. Combining all the latter we get that for all bids $i \geq 1$,
\begin{eqnarray*}
\sum_{j \in C_k} \sigma_B^k(j) \cdot x_A(0,j)\cdot (p_A(0,j) - \epsilon_k) &\geq& \sum_{j \in M} \sigma_B^k(j) \cdot x_A(0,j)\cdot (p_A(0,j) - \epsilon_k)\\
&\geq& \sum_{j \in M} \sigma_B^k(j) \cdot x_A(i,j)\cdot (p_A(i,j) - \epsilon_k)\\
&\geq& \sum_{j \in C_k} \sigma_B^k(j) \cdot x_A(i,j)\cdot (p_A(i,j) - \epsilon_k)
\end{eqnarray*}
By dividing the latter inequality with $1/(1 - \sum_{\ell \notin C_k}\sigma_B(\ell))$ we get that  
\begin{eqnarray*}
\sum_{j \in C_k} \hat{\sigma}_B^k(j) \cdot x_A(0,j)\cdot (p_A(0,j) - \epsilon_k) \geq\sum_{j \in C_k} \hat{\sigma}_B^k(j) \cdot x_A(i,j)\cdot (p_A(i,j) - \epsilon_k)
\end{eqnarray*}
meaning that agent $A$ has no exploitability against strategy $\hat{\sigma}_B$. At the same time, we are ensured that the payment of agent $A$ is at least $\mu_A := \min_{i,j \in [M]} \{x_{A}(i,j)\cdot p_A(i,j)~:~ x_{A}(i,j)\cdot p_A(i,j) > 0\}$.
\smallskip
\smallskip

\noindent We complete the proof by showing that the exploitability agent $B$ with respect to $e_0 = (1,0,\ldots,0)$ is at most $M\cdot(1-\sigma_A^k(0))$. This is because
\[ \left|\sum_{i \in [M]} \sigma_A(i) \cdot x_A(i,j)\cdot p_A(i,j)  -  \sum_{i \in [M]} e_0(i) \cdot x_A(i,j)\cdot p_A(i,j)\right| \leq M \cdot (1-\sigma_A^k(0))\]
\smallskip
\smallskip

We complete the proof of Lemma~\ref{l:main} in case agent $A$ does not sell any of its energy (the amount of energy that agent $B$ sells is  $1$). This means that for any bid $i \in [M]$, $x_A(i,j) = 0$ for all $j \in \mathrm{support}(\sigma_B^k)$ since otherwise agent $A$ could deviate to a strategy and sell some of its energy making its revenue positive. Now let $j_\mathrm{max} := \mathrm{argmax}_{j \in \mathrm{support(\sigma_B)}} \sum_{ i \in [M]} \sigma_A(i) \cdot p_B(i,j)$. Then the pair of strategies $\hat{\sigma} =  ( \sigma_A, e_{j_\mathrm{max}})$ is a mixed NE
since $x_A(i,j_{\mathrm{max}}) = 0$ for all $i \in [M]$. At the same time the payment to agent $B$ is at least $1$ since $\mathrm{max}_{j \in \mathrm{support(\sigma_B)}} \sum_{ i \in [M]} \sigma_A(i) \cdot p_B(i,j) \geq 1$.
\qed

\smallskip
\smallskip
\smallskip

\item \underline{There exists an infinite subsequence $\mathcal{S}' \subseteq 
\mathcal{S}$ such that for any $\sigma_k \in \mathcal{S}'$,  $\mathrm{support}(\sigma^k_A) = \{0\}$}.
\smallskip
\smallskip

In this case, we are ensured that at limiting point $\sigma^\star$, $\lim_{k \rightarrow \infty} \sigma^k = \sigma^\star$ and the proof follows with the exact same steps as above.

\end{itemize}
\qed
\end{proof}


\section{Proof of Theorem~\ref{thm:UB-PC}}
\label{proof:thm:UB-PC}

\TUBPC*

\begin{proof}
Let agent $i$ be the agent that has $\high{b}_i=\max_{j\preceq_\mathbf{c}\tau(\mathbf{c})}\high{b}_j$ and let $\mathbf{b}^\star_{-i}=\mathbf{c}_{-i}+\mathbf{d}_{-i}$ where $\mathbf{d}_{-i}\in\set{0,1}^{n-1}$ such that $\high{b}_i=\max\set{b_i\in[M] \text{ s.t. }b_i\in \mathcal{BR}(\mathbf{b}^\star_{-i})}$. We will establish that $\vecb^\star=(\high{b}_i,\vecb^\star_{-i})$ is a pure NE and $p^{unit}(\vecb^\star)=\high{b}_i$.
\smallskip

\noindent By the definition of $\vecb^\star$ agent $i \in [n]$ has no incentive to deviate. Now consider any other agent $j\neq i$. 

\begin{itemize}
\item Let agent $j\preceq_\mathbf{c}\mathbf{c}$, for any bid $b_j\leq \high{b}_i$, the agent $j$ will sell less or equal amount of energy at price $\high{b}_i$. Thus, $U_j(b_j,b^\star_{-j}) \leq U_j(c_j,b^\star_{-j})$. 
As a result, we need to consider only the case $b_j > \high{b}_i$. By the definition, we know that 
\begin{align*}
    \max\set{b_j\in[M]\text{ s.t. }b_j\in\mathcal{BR}(c_i,\vecb^\star_{-i-j})}\leq \high{b}_j \leq \high{b}_i,
\end{align*} while in the bidding profile $\vecb^\star$ agent $j \in [n]$ sells $s_j$ amount of energy. This means that \[\max_{b_j\in[M]}U_j(b_j,(c_i,\mathbf{b^\star}_{-i-j}))\leq (\high{b}_i-c_j)\cdot s_j = U_j(\vecb^\star). \] Therefore, we have  
\[U_j(b_j, \vecb^\star_{-j}) = U_j(b_j, (c_i,\mathbf{b^\star}_{-i-j})) \leq \max_{b_j\in[M]}U_j(c_j,(b_i,\mathbf{b^\star}_{-i-j}))  \leq U_j(\vecb^\star).\]
The equality is due to $b_j>\high{b}_i=b^\star_i$. Thus, agent $j$ has no incentive to deviate.
\smallskip
\smallskip
\smallskip

\item  In the case $j \succ_\mathbf{c}\tau(\mathbf{c})$, if $c_j<\high{b}_i$, the same argument as above follows. If $c_j \geq \high{b}_i$, agent $j$ sells $0$ amount of energy for any bid $b_j \geq \high{b}_i$. This is because $\sum_{j\preceq_\mathbf{c}\tau(\mathbf{c})} s_j \geq 1$. For any bid $b_j < \high{b}_i$, the agent with $c_j \geq \high{b}_i$ cannot have non-positive utility. In any case, an agent $j$ does not have an incentive to deviate.

\noindent Now, we consider the unit price $p(\vecb^\star)$. By the maximality of $\high{b}_i$, $p(\vecb^\star)$ cannot be greater than $\high{b}_i$. At the same time, $p(\vecb^\star)$ is at least $\high{b}_i$ since agent $i$ yields a positive utility in the bidding profile $\mathbf{c}+1$ where every agent bids the real cost plus one, which implies $\high{b}_i> c_i$. Thus, $p(\vecb^\star)=\high{b}$.
\end{itemize}
\qed
\end{proof}

\section{Proof of Lemma \ref{lem:least-utility-PC}}
\label{proof:lem:least-utility-PC}
\LleasttuilityPC*

\begin{proof}

     The lemma follows since
     \begin{align}
     \label{inq:lem-pc-least-utility}
         U_i(\mathbf{b})= x_i(\mathbf{b})\cdot(q(\vecb)-c_i)\geq x_i(b_i,\mathbf{c}_{-i})\cdot(q(b_i,\mathbf{c}_{-i})-c_i)=U_i(b_i,\mathbf{c}_{-i}).
     \end{align}
     The equalities hold by definition. The inequality holds since $x_i(\mathbf{v})$ and $q(\mathbf{v})$ are monotone with respect to $v_j$ where $j\neq i$ and $b_j\geq c_j$ holds for any $j\neq i$.  
     
     \smallskip
     \noindent If we let $b^\star_i= \max \mathcal{BR}(\mathbf{c}_{-i})$, and consider any mixed NE $\sigma$, we have
     \begin{align*}
         \Ex_{\vecb\sim\sigma}[U_i(\vecb)]&\geq \max_{b'_i\in[M]}\Ex_{\vecb_{-i}\sim\sigma_{-i}}[U_i(b'_i,\vecb_{-i})]\geq \Ex_{\vecb_{-i}\sim\sigma_{-i}}[U_i(b^\star_i,\vecb_{-i})]\\
         &=\sum_{\vecb_{-i}\in \text{support}(\sigma_{-i})}\Pr[\vecb_{-i}]\cdot U_i(b^\star_i,\vecb_{-i})\geq \sum_{\vecb_{-i}\in \text{support}(\sigma_{-i})}\Pr[\vecb_{-i}]\cdot U_i(b^\star_i,\mathbf{c}_{-i})\\
         &=U_i(b^\star_i,\mathbf{c}_{-i})=\max_{b'_i\in[M]} U_i(b'_i,\mathbf{c}_{-i}).
     \end{align*}
     The first inequality is due to the equilibrium condition. The second inequality follows from the maximality. The first equality is by definition. The third inequality follows from Inequality \eqref{inq:lem-pc-least-utility}. The second equality follows by the normalization of probabilities. The last equality holds due to our choice of $b^\star_i$. \qed
\end{proof}
\section{Proof of Theorem \ref{thm:LB-PC}}
\label{proof:thm:LB-PC}

\TLBPC*
\begin{proof}
Assume there is a mixed NE $\sigma$ with $\Ex_{\vecb\sim\sigma}[p^{\sss \text{unit}}(\vecb)]<\max_{i\preceq_\mathbf{c}\tau(\mathbf{c})}\low{b}_i$. It suggest the expected clearing price $\Ex_{\vecb\sim\sigma}[q(\vecb)]<\max_{i\preceq_\mathbf{c}\tau(\mathbf{c})}\low{b}_i$ since 
\begin{align*}
    p^{\sss \text{unit}}(\vecb)=\sum_{i \in [n]}x_i(\vecb) \cdot p_i(\vecb)=\sum_{i \in [n]}x_i(\vecb) \cdot q(\vecb)=q(\vecb),
\end{align*}
where all equalities are by definition.

\noindent Fix an agent $j$ whose $\low{b}_j=\max_{i\preceq_\mathbf{c}\tau(\mathbf{c})}\low{b}_i$, we have 
\begin{align*}
    \Ex_{\vecb\sim\sigma}[U_j(\vecb)]
    &=\Ex_{\vecb\sim\sigma}[x_j(\vecb)\cdot(p_j(\vecb)-c_j)]=\Ex_{\vecb\sim\sigma}[x_j(\vecb)\cdot(q(\vecb)-c_j)]\\
    &\leq \Ex_{\vecb\sim\sigma}[s_j\cdot(q(\vecb)-c_j)]=s_j\cdot(\Ex_{\vecb\sim\sigma}[q(\vecb)]-c_j)\\
    &<s_j\cdot(\max_{i\preceq_\mathbf{c}\tau(\mathbf{c})}\low{b}_i-c_j)=s_j\cdot(\low{b}_j-c_i)=\max_{b'_i\in[M]}U_j(b'_j,\mathbf{c}_{-j}).
\end{align*}
The first and second equality follow by definition. The first inequality holds since $s_j$ is the maximum allocation agent $j$ might receive. The third equality holds by the linearity of expectation. The second inequality follows since we have shown $\Ex_{\vecb\sim\sigma}[q(\vecb)]<\max_{i\preceq_\mathbf{c}\tau(\mathbf{c})}\low{b}_i$. The fourth equality holds by our choice of agent $j$. The last equality follows by the definition of $\low{b}_j$.

\noindent However, Lemma \ref{lem:least-utility-PC} guarantees that 
\begin{align*}
    \Ex_{\vecb\sim\sigma}[U_j(\vecb)]\geq \max_{b'_j\in[M]}U_j(b'_j,\mathbf{c}_{-j}),
\end{align*}
which leads to a contradiction.
\qed
\end{proof}

\section{Proof of Lemma \ref{lem:least-utility-PB}}
\LleasttuilityPB*
\label{proof:lem:least-utility-PB}
\begin{proof}
The proof is almost identical to the proof of Lemma \ref{lem:least-utility-PC}. We have
\begin{align}
\label{inq:lem-pb-least-utility}
    U_i(\mathbf{b})=x_i(\mathbf{b})\cdot(b_i-c_i)\geq x_i(b_i,\mathbf{b}'_{-i})\cdot(b_i-c_i)=U_i(b_i,\mathbf{b}'_{-i}).
\end{align}
The equalities hold by definition. The inequality holds since $x_i(\mathbf{v})$ are monotone with respect to $v_j$ where $j\neq i$ and $b_j\geq b'_j$ holds for any $j\neq i$. 

\noindent If we let $b^\star_i= \max \mathcal{BR}(\mathbf{c}_{-i})$, and consider any mixed NE $\sigma$, we have
     \begin{align*}
         \Ex_{\vecb\sim\sigma}[U_i(\vecb)]&\geq \max_{b'_i\in[M]}\Ex_{\vecb_{-i}\sim\sigma_{-i}}[U_i(b'_i,\vecb_{-i})]\geq \Ex_{\vecb_{-i}\sim\sigma_{-i}}[U_i(b^\star_i,\vecb_{-i})]\\
         &=\sum_{\vecb_{-i}\in \text{support}(\sigma_{-i})}\Pr[\vecb_{-i}]\cdot U_i(b^\star_i,\vecb_{-i})\geq \sum_{\vecb_{-i}\in \text{support}(\sigma_{-i})}\Pr[\vecb_{-i}]\cdot U_i(b^\star_i,\mathbf{c}_{-i})\\
         &=U_i(b^\star_i,\mathbf{c}_{-i})=\max_{b_i\in[M]} U_i(b_i,\mathbf{c}_{-i}).
     \end{align*}
The first inequality is due to the equilibrium condition. The second inequality follows from the maximality. The first equality is by definition. The third inequality follows from Inequality \eqref{inq:lem-pb-least-utility}. The second equality follows by the normalization of probabilities. The last equality holds due to our choice of $b^\star_i$. 
\qed
\end{proof}

\section{Proof of Theorem \ref{thm:LB-PB}}
\label{proof:thm:LB-PB}
\TLBPB*
\begin{proof}
Fix the agent $i:=\argmax_{j\preceq_{\mathbf{c}}\tau(\mathbf{c})}\low{b}_j$. Due to Lemma~\ref{lem:least-utility-PB} and the equilibrium condition, any bid in agent $i$'s support should yield a utility of at least $\max_{b'_i\in[M]}U_i(b'_i,\mathbf{c}_{-i})$. Recall $\low{b}_i= c_i + \max_{b'_i\in[M]}U_i(b'_i,\mathbf{c}_{-i}) / s_i$, agent $i$ has no incentive to bid below $ \low{b}_i$ with non-zero probability in $\sigma$. Consider any bid $b_i < \low{b}_i$, even in case agent $i \in [n]$ sells all of its energy $s_i \in [0,1]$, its utility is strictly less than $\max_{b_i\in[M]}U_i(b_i,\mathbf{c}_{-i})$. As a result, we are ensured that the with probability $1$, agent $i \in [n]$ bids greater or equal to $\low{b}_i:= \max_{j\preceq_{\mathbf{c}}\tau(\mathbf{c})}\low{b}_j$.
\smallskip
\smallskip

\noindent Notice that since agent $i \in [n]$ sells a positive amount of energy once bidding $\high{b}_i$ and since $\low{b}_i \leq \high{b}_i$, we are ensured that $\sum_{ c_j < \low{b}_i}s_j < 1$. Therefore, any agent $j\neq i$ with cost $c_j < \low{b}_i$, is ensured to sell his whole supply $s_j$ by bidding $\low{b}_i-1$. In this case, bidding lower cannot increase the allocation. Thus, any agent $j \neq i$ will have a strictly smaller payoff once bidding less than $\low{b}_i-1$. \qed
\end{proof}

\section{Proof of Theorem~\ref{thm:UB-PB}}\label{proof:thm:UB-PB}

\TUBPB*

\noindent The proof of Theorem~\ref{thm:UB-PB} follows easily by Lemma~\ref{lem:ub-low-cost-agent} which is presented up next. The proof of Lemma~\ref{lem:ub-low-cost-agent} is presented in Appendix~\ref{proof:lem:ub-low-cost-agent}. 
\begin{lemma}
    \label{lem:ub-low-cost-agent}
    Let $t:=\max_{i\preceq_\mathbf{c}\tau(\mathbf{c})}\max_{\sigma_i(b_i)>0} b_i$ be the maximum possible bid of an agent $i\preceq_\mathbf{c}\tau(\mathbf{c})$. Let agent $i\preceq_\mathbf{c}\tau(\mathbf{c})$ be the agent with the lowest priority according to the lexicographical order among agents who bid $t$ with positive probability, we have $t\leq \high{b}_i$.
\end{lemma}
To this end we complete the section with the proof of Theorem~\ref{thm:UB-PB}.
\begin{proof}[Proof of Theorem~\ref{thm:UB-PB}]
    By Lemma \ref{lem:ub-low-cost-agent}, the highest possible bid from all the agent $i\preceq_\mathbf{c}\tau(i)$ is at most $\max_{i\preceq_\mathbf{c}\tau(\mathbf{c})}\high{b}_i$. Since $\sum_{i\preceq_\mathbf{c} \tau(\mathbf{c})}s_i\geq 1$, any agent $j\succ_\mathbf{c}\tau(i)$ cannot sell positive amount by bidding $b_j\geq\max_{i\preceq_\mathbf{c}\tau(\mathbf{c})}\high{b}_i$. Thus, no energy is bought at price higher than $\max_{i\preceq_\mathbf{c}\tau(\mathbf{c})}\high{b}_i$.\qed
\end{proof}

\subsection{Proof of Lemma \ref{lem:ub-low-cost-agent}}
\label{proof:lem:ub-low-cost-agent}
 Due to Lemma \ref{lem:least-utility-PB}, agent $i$'s utility is at least $\max_{b_i\in[M]}U_i(b_i,\mathbf{c}_{-i})$. Thus, $\Ex_{\vecb_{-i}\sim\sigma_{-i}}[U_i(t,\vecb_i)]$ is at least $\max_{b_i\in[M]}U_i(b_i,\mathbf{c}_{-i})$ by equilibrium condition. We first assume that $\max_{b_i\in[M]}U_i(b_i,\mathbf{c}_{-i})>0$ and will solve the case $\max_{b_i\in[M]}U_i(b_i,\mathbf{c}_{-i})=0$ later. Under our assumption, we notice a specific behavior by agents $j\succ_\mathbf{c}\tau(\mathbf{c})$ with cost $c_j<t-1$.

\begin{claim2}
    \label{clm:low-cost-agent-always-bids-low}
    If an agent $j\succ_\mathbf{c}\tau(\mathbf{c})$ has cost $c_j<t-1$, it will bid at most $t-1$ with probability of $1$.
\end{claim2}

\begin{proof}
    Since $\Ex_{\vecb_{-i}\sim\sigma_{-i}}[U_i(t,\vecb_i)]\geq\max_{b_i\in[M]}U_i(b_i,\mathbf{c}_{-i}) >0$, we notice there exists at least one bidding profile $\vecb\sim\sigma$ where $b_i=t$ such that $\tau(\mathbf{b})\succeq_\vecb i$. Together with the fact $\sum_{k\preceq_\mathbf{c}\tau(\mathbf{c})}s_k\geq 1$ and our choice of $i$, $\tau(\mathbf{b})= i$ in such $\vecb$. Thus, any agent $j\succ_\mathbf{c}\tau(\mathbf{c})$ with $c_j<t-1$ will always bid at most $t-1$. Otherwise, it sells an amount of $0$ with probability of $1$ due to the fact $\sum_{k\preceq_\mathbf{c}\tau(\mathbf{c})}s_k\geq 1$ and our lexicographical tie-breaking.
    \qed
\end{proof}

\noindent Next, we establish the value of $t$ in two complementary cases.
\begin{itemize}
    \item {\bf Case 1:} There are no agent $j\succ_\mathbf{c}\tau(\mathbf{c})$ has cost $c_j=t-1$.
    
    \smallskip
    \noindent In this case, we have
    \begin{align}
        \label{inq:case1-ub}
        \Ex_{\vecb_{-i}\sim\sigma_{-i}}[U_i(t,\vecb_{-i})]=U_i(t,\mathbf{c}_{-i})\leq\max_{b_i\in[M]}U_i(b_i,\mathbf{c}_{-i}).
    \end{align}
    The equality holds due to our choice of $i$ and Claim \ref{clm:low-cost-agent-always-bids-low}, since $c_j<t$ implies $c_j<t-1$ in this case. And the inequality is due to the maximality. Furthermore, we have 
    \begin{align}
        \label{inq:case1-lb}
        \Ex_{\vecb_{-i}\sim\sigma_{-i}}[U_i(t,\vecb_{-i})]=\Ex_{\vecb\sim\sigma}[U_i(\vecb)]\geq \max_{b_i\in[M]}U_i(b_i,\mathbf{c}_{-i}).
    \end{align}
    The equality follows from equilibrium condition. And the inequality is by Lemma \ref{lem:least-utility-PB} and linearity of expectation. 

    Putting Inequality \eqref{inq:case1-ub} and $\eqref{inq:case1-lb}$ together, it is clear that $t$ is a best response to $\mathbf{c}_{-i}$. Thus, \begin{align*}
        t\leq \max\set{b_i\in[M]\text{ such that } b_i\in\mathcal{BR}(\mathbf{c}_{-i})}\leq\high{b}_i.
    \end{align*}
    \item {\bf Case 2:} If there is at least one agent $j\succ_\mathbf{c}\tau(\mathbf{c})$ has cost $c_j=t-1$.

    \smallskip
    \noindent In this case, an agent $j\succ_\mathbf{c}\tau(\mathbf{c})$ with $c_j=t-1$ might bid anything greater than or equal to $c_k$. However, when agent $i$ bids $t$, the conditional expected utility depends only on whether every agent $j$ bids $c_j=t-1$ or greater than $c_j$, in other words, at least $t$. Thus, agent $i$'s expected utility conditioned on bidding $t$ is 
    \begin{align*}
        \Ex_{\vecb_{-i}\sim\sigma_{-i}}[U_i(t,\vecb_{-i})]=\sum_{\mathbf{d}_{-i}\in\set{0,1}^{n-1}} \Pr_{\vecb_{-i}\sim\sigma_{-i}}\left[\bigwedge_{j\neq i}\min\{b_j,c_j+1\}=c_j+d_j\right]\cdot U_i(t,\mathbf{c}_{-i}+\mathbf{d}_{-i})
    \end{align*}
    Note that $\Pr_{\vecb_{-i}\sim\sigma_{-i}}\left[\bigwedge_{j\neq i}\min\{b_j,c_j+1\}=c_j+d_j\right]$ denotes the probability that every agent $j$ bids $b_j=c_j$ if $d_j=0$, otherwise, bids $b_j\geq c_j+1$. We define $\mathbf{d}^\star_{-i}\in\set{0,1}^{n-1}$ where $d^\star_{j}=1$ if and only if $c_j+1=t$. We have
    \begin{align}
        \label{inq:case2-PB-ub}
        \Ex_{\vecb_{-i}\sim\sigma_{-i}}[U_i(t,\vecb_{-i})]\leq U(t,\mathbf{c}_{-i}+\mathbf{d}^\star_{-i}).
    \end{align}
    The inequality follows since $U_i(t,\mathbf{c}_{-i}+\mathbf{d}_{-i})$ is maximized when $\mathbf{d}_{-i}=\mathbf{d}^\star_{-i}$. To understand, $\mathbf{d}^\star_{-i}$ says every agent $j$ with $c_j=t-1$ will bid $t$ so that the allocation $x_i(t,\mathbf{c}_{-i}+\mathbf{d}_{-i})$ is maximized when $\mathbf{d}_{-i}=\mathbf{d}^\star_{-i}$. 
    Moreover, for any $b_i\neq t$, we have
    \begin{align}
        \label{inq:case2-PB-lb}
        \Ex_{\vecb_{-i}\sim\sigma_{-i}}[U_i(t,\vecb_{-i})]\geq \max_{b_i\neq t}\Ex_{\vecb_{-i}\sim\sigma_{-i}}[U_i(b_i,\vecb_{-i})]\geq \max_{b_i\neq t} U_i(b_i,\mathbf{c}_{-i}+\mathbf{d}^\star_{-i}). 
    \end{align}
    The first inequality is by equilibrium condition.
    The reason why the second inequality follows is: 1. Agent $i$'s utility while fixing a specific bid is minimized when the allocation is minimized; 2. We have $b_i\neq t$, so the agent $j$ with $c_j=t-1$ does not affect the utility if it bids $c_j+1$.
    
    Putting Inequality \eqref{inq:case2-PB-ub} and \eqref{inq:case2-PB-lb} together, we get 
    \begin{align*}
        U(t,\mathbf{c}_{-i}+\mathbf{d}^\star_{-i})=\max_{b_i\in[M]}U_i(b_i,\mathbf{c}_{-i}+\mathbf{d}^\star_{-i}).
    \end{align*}
    Thus, by definition,
    \begin{align*}
        t&\leq\max\set{b_i\in[M] \text{ such that } b_i\in\mathcal{BR}(\mathbf{c}_{-i}+\mathbf{d}^\star_{-i})}\\
        &\leq \max_{\mathbf{d}\in\set{0,1}^n}\max\set{b_i\in[M] \text{ such that } b_i\in\mathcal{BR}(\mathbf{c}_{-i}+\mathbf{d}_{-i})}=\high{b}_i.
    \end{align*}
\end{itemize}

\smallskip
\smallskip

\noindent Finally, we consider what if $\max_{b_i\in[M]}U_i(b_i,\mathbf{c}_{-i})=0$. In this case, we have $\high{b}_i=c_i+1$. Assume $t> c_i+1$, agent $j$ with $c_j=c_i$ will bid in $[c_i+1, t-1]$ with probability of $1$ since $\sum_{k\preceq_\mathbf{c} \tau(\mathbf{c})}s_k\geq 1$. And when agent $i$ bids $t$, it receives a conditional expected utility of $0$ since $\sum_{k\neq i,c_k\leq c_i}s_k\geq 1$. It contradicts that $\sigma$ is a mixed NE since agent $i$ can move all the probability mass of bidding $t$ to bid $c_i+1$ to improve the utility. 

\noindent In all cases, the theorem follows.
\qed

\section{Proof of Theorem \ref{thm:UB-PB-better}}
\label{proof:thm:UB-PB-better}
\TUBPBbetter*

\begin{proof} 
Let $i^\star:=\argmax_{i\preceq_\mathbf{c}\tau(c)}\low{b}_i$. Assume that if $s_{i^\star}<1$, $\low{b}_{i^\star}+1< \high{b}_{i^\star}$, and $\low{b}_j\leq\low{b}_{i^\star}-2$ for every agent $j\neq i^\star$. We will show that the unit price given by the worst-case mixed NE of PB is at most $p:= (1-s_{i^\star})\cdot\left(\low{b}_{i^\star}+1+\gamma\cdot\sum_{\ell=\low{b}_{i^\star}+2}^{\high{b}_{i^\star}}{1}/{(\ell-c_{i^\star}-1)}\right)+s_{i^\star}\cdot\high{b}_{i^\star}$ where $\gamma=\max_{b'_{i^\star}\in[M]}U_{i^\star}(b'_{i^\star},\mathbf{c}_{-i^\star})/s_{i^\star}+1$.  Moreover, $p < \high{b}_{i^\star}\leq \max_{i\preceq_{\mathbf{c}}\tau(\mathbf{c})}\high{b}_i$.

\noindent We remark that $s_{i^\star} < 1$ is a natural assumption since, most likely, no single producer $i \in [n]$ will be able to cover the whole energy demand. The assumption $\low{b}_{i^\star}+1<\high{b}_{i^\star}$, is also mild since $\low{b}_{i^\star}+1\geq \high{b}_{i^\star}$ only happens when $i \in [n]$ sells exactly $s_{i^\star}$ energy as a pivotal agent. This only happens if $\sum_{j: c_j \leq \high{b}_{i^\star} }  s_j$ is equal to the total energy demand of $1$. This case will occur with $0$ probability assuming a small random perturbation on the  supplies $s_j$ of the agents. Similarly, assuming in case of sufficiently small discretization and a random permutation on the  supplies $s_j$, the assumption $\low{b}_j\leq\low{b}_{i^\star}-2$ is satisfied with probability $1$.

\noindent We start by tightening the support of the mixed NE.
\begin{lemma}
\label{lem:tighten-supp}
Let $i^\star:=\argmax_{i\preceq_\mathbf{c}\tau(c)}\low{b}_i$. In case that every agent $j\neq i^\star$ has $\low{b}_j\leq\low{b}_{i^\star}-2$. Then the support of any mixed NE is a subset of $[\low{b}_{i^\star}-1,\high{b}_{i^\star}]$. 
\end{lemma}

\begin{proof}
The lower bound follows directly from Theorem \ref{thm:LB-PB}, so we only need to prove the upper bound. We first establish that in case agent $j\in [n]$ sells its supply $s_j \in [0,1]$ at price $\low{b}_j+1$, then its utility is higher than the \textit{best possible utility} once bidding $\high{b}_j$. By the definition of $\low{b}_j$ we have that,
    \begin{align}
        \label{inq:tighter_interval_1}
        \low{b}_j=\left\lceil c_j+\frac{\max_{b_j\in[M]}U_j(b_j,\mathbf{c}_{-j})}{s_j}\right\rceil\geq c_j+\frac{\max_{b_j\in[M]}U_j(b_j,\mathbf{c}_{-j})}{s_j}.
    \end{align}
    Then we consider the maximum utility achieved by becoming the pivotal agent, we have that
    \begin{align}
        \nonumber
        \max_{\mathbf{d}_{-j}\in\set{0,1}^{n-1}}\max_{b_j\in[M]}U_j(b_j,\mathbf{c}_{-j}+\mathbf{d}_{-j})&\leq \max_{b_j\in[M]}U_j(b_j,\mathbf{c}_{-j}+1)\\
        \nonumber
        &\leq\max_{b_j\in[M]}U_j(b_j-1,\mathbf{c}_{-j})+(b_j-(b_j-1))\cdot s_j\\ 
          \label{inq:tighter_interval_2}          &=\max_{b_j\in[M]}U_j(b_j,\mathbf{c}_{-j})+s_j.
    \end{align}
    The first inequality follows since the price under the same allocation is maximized if every other agent bids one step to the left. The second inequality follows since the allocation is the same if every agent shifts its bid one step to the left. The equality is simply by algebra.
    
    \smallskip
    \noindent Plug Inequality \eqref{inq:tighter_interval_2} into \eqref{inq:tighter_interval_1}, we get
    \begin{align*}
        \low{b}_j+1\geq c_j+\frac{\max_{\mathbf{d}_{-j}\in\set{0,1}^{n-1}}\max_{b_j\in[M]}U_j(b_j,\mathbf{c}_{-j}+\mathbf{d}_{-j})}{s_j}.
    \end{align*}

\noindent Notice that by Theorem~\ref{thm:LB-PB}, the support of any mixed NE is greater than $\low{b}_{i^\star} - 1$ meaning that if agent $j$ bids $\low{b}_{i^\star} - 2 \geq \low{b}_j + 1$ then agent $j \in [n]$ sells its whole energy $s_j$ at price at least $\low{b}_j +1$, meaning that the utility of bidding $\low{b}_{i^\star} - 2$ is at least $U_j(\high{b}_j,\mathbf{c}_{-j}+\mathbf{d}_{-j})$.

\noindent Thus, the agent $j$ with $\low{b}_j\leq\low{b}_j-2$ will not have a positive probability to bidding any value to make it become the pivotal agent surely since the utility is always better if it can sell all the supply with price $\low{b}_j+1$. And by a similar argument as in the Theorem \ref{thm:LB-PB}, agent $i^\star$ bids at least $\low{b}_{i^\star}$ with probability of $1$. Thus, agent $j$ is ensured to sell everything by bidding $\low{b}_{i^\star}-1$. Recall we assume that $\low{b}_j\leq \low{b}_{i^\star}-2$. So any agent $j\neq i$ prefers to bid $\low{b}_{i^\star}-1$ than becoming the pivotal agent, in other words, it cannot be the one who bids the highest possible bid. The lemma follows since we can repeat the proof of Theorem \ref{thm:UB-PB} with the knowledge that agent $i$ bids the highest.
    \qed
\end{proof}
Now, consider a mixed NE $\sigma$. Assume that $t=\max_{\sigma_{i^\star}(b_{i^\star})>0}b_{i^\star}$. By Lemma \ref{lem:tighten-supp}, we have $\sigma_{i^\star}(b_{i^\star})>0$ only if $\low{b}_{i^\star}-1\leq b_{i^\star} \leq \high{b}_{i^\star}$. Assume without loss of generality that $\sigma_{i^\star}(\high{b}_{i^\star})>0$; otherwise, the unit price is always lower than the bound we present. Due to the equilibrium, agent $i^\star$'s expected utility conditioned on bidding $\high{b}_{i^\star}$ is at least the expected utility conditioned on bidding any other $b_{i^\star}\in [M]$, i.e.,
    \begin{align} 
        \label{inq:unit-price-pb-1}
        \Ex_{\vecb_{-i^\star}\sim\sigma_{-i^\star}}[U_{i^\star}(\high{b}_{i^\star},\vecb_{-i^\star})]&\geq \Ex_{\vecb_{-i^\star}\sim\sigma_{-i^\star}}[U_{i^\star}(b_{i^\star},\vecb_{-i^\star})]\geq s_{i^\star}\cdot (b_{i^\star}-c_{i^\star})\cdot\Pr_{\vecb_{-i\star}\sim \sigma_{-i}}[x_{i^\star}(b_{i^\star},\vecb_{-i^\star})=s_{i^\star}]
    \end{align}
    The second inequality follows since we ignore any utility given by the case $x_{i^\star}(b_{i^\star},\mathbf{b}_{-i^\star})\neq s_{i^\star}$, in other words, agent $i^\star$ sells all its supply while bidding $b_{i^\star}$. Furthermore, 
    \begin{align}
        \label{inq:unit-price-pb-2}
        \Ex_{\vecb_{-i}\sim\sigma_{-i}}[U_{i^\star}(\high{b}_{i^\star},\vecb_{-i^\star})]\leq \max_{\mathbf{d}_{-{i^\star}}\in\set{0,1}^{n-1}}\max_{b_{i^\star}\in[M]}U_{i^\star}(b_i,\mathbf{c}_{-i^\star}+\mathbf{d}_{-i^\star})\leq \max_{b_{i^\star}\in[M]}U_{i^\star}(b_{i^\star},\mathbf{c}_{-i^\star})+s_{i^\star}.
    \end{align}
    The first inequality follows since $\high{b}_{i^\star}$ is either $c_j$ or $c_j+1$ for some agent $j$. The second inequality follows by Inequality \eqref{inq:tighter_interval_2}.
    Puttng Inequality \eqref{inq:unit-price-pb-1} and \eqref{inq:unit-price-pb-2} together, we have
    \begin{align*}
        \Pr_{\vecb_{-i}\sim\sigma_{-i}}[x_{i^\star}(b_{i^\star},\vecb_{-i^\star})=s_{i^\star}]\leq \frac{\max_{b'_{i^\star}\in[M]}U_{i^\star}(b'_{i^\star},\mathbf{c}_{-i^\star})+s_{i^\star}}{s_{i^\star}\cdot(b_{i^\star}-c_{i^\star})}=\frac{\gamma}{b_{i^\star}-c_{i^\star}}.
    \end{align*}
\smallskip
\smallskip

\noindent Notice that in case $x_{i^\star}(b_{i^\star},\mathbf{b}_{-i^\star})<s_{i^\star}$ then at least $1-s_{i^\star}$ energy is bought at price at most $b_{i^\star}$. Denote the random variable $X_{\leq \alpha}$ as the amount of energy that is bought at a price at most $\alpha$. The randomness comes from which bidding profile is drawn from the distribution. We have
    \begin{align}
        \label{inq:unit-price-pb-prob}
        \Pr_{\vecb\sim\sigma}[X_{\leq\alpha}\geq 1- s_{i^\star}]\geq 1-\Pr_{\vecb_{-i}\sim\sigma_{-i}}[x_{i^\star}(b_{i^\star}=\alpha,\vecb_{-i^\star})=s_{i^\star}]\geq 1-\frac{\gamma}{\alpha-c_{i^\star}}.
    \end{align}

    \noindent We can then upperbound the unit price by minimizing the energy sold at a low price. Specifically, when $X_{\leq \alpha}\geq 1-s_{i^\star}$ but $X_{\leq \alpha-1}< 1-s_{i^\star}$, we assume exactly $1-s_{i^\star}$ energy is sold below the price $\alpha$ and the remaining $s_{i^\star}$ energy is sold at the price of $\high{b}_{i^\star}$. In this way, we will only increase the unit price. As a result,
    \begin{align}
        \label{inq:unit-price-relaxation-1}
        \Ex_{\vecb\sim\sigma}[p^{unit}(\vecb)]\leq \sum_{\alpha=\low{b}_{i^\star}-1}^{\high{b}_{i^\star}}f(\alpha)\cdot (s_{i^\star}\cdot\high{b}_{i^\star}+(1-s_{i^\star})\cdot \alpha),
    \end{align}
    where we denote \[f(\alpha)=\Pr_{\vecb\sim\sigma}[X_{\leq \alpha}\geq 1-s_{i^\star}\text{ and }X_{\leq \alpha-1}< 1-s_{i^\star}]=\Pr_{\vecb\sim\sigma}[X_{\leq \alpha}\geq 1-s_{i^\star}]-\Pr_{\vecb\sim\sigma}[X_{\leq \alpha-1}\geq 1-s_{i^\star}].\]

    \noindent Notice that $\Pr_{\vecb\sim\sigma}[X_{\leq \alpha}\geq 1-s_{i^\star}]=0$ for any $\alpha < \low{b}_{i^\star}-1$ since every agent bids at least $\low{b}_{i^\star}$. And $\Pr_{\vecb\sim\sigma}[X_{\leq \alpha}\geq 1- s_{i^\star}]=1$ for any $\alpha \geq\high{b}_{i^\star}$ since every agent bids at most $\high{b}_{i^\star}$. Thus, we can upperbound the RHS of Inequality \eqref{inq:unit-price-relaxation-1} by as setting,

\smallskip

    \begin{align}
    \label{inq:unit-price-tight-prob-est}
        f(\alpha)=\begin{cases}
            0, &\alpha=\low{b}_{i^\star}-1 \text{ or }\low{b}_{i^\star},\\
            1-\frac{\gamma}{\alpha-c_{i^\star}},&\alpha=\low{b}_{i^\star}+1,\\
            \gamma\cdot(\frac{1}{\alpha-c_{i^\star}-1}-\frac{1}{\alpha-c_{i^\star}}),&\low{b}_{i^\star}+2\leq \alpha\leq\high{b}_{i^\star}-1\\
            \gamma\cdot\frac{1}{\alpha-c_{i^\star}-1}, &\alpha=\high{b}_{i^\star}\\
        \end{cases}.
    \end{align}
    
\noindent This is because the assignment presenting in Equation~\eqref{inq:unit-price-tight-prob-est} is the assignment maximizing Equation~\eqref{inq:unit-price-relaxation-1} while at the same time satifying Equation~\eqref{inq:unit-price-pb-prob} and the non-negativity of probabilities. 
\noindent Plugging Equality \eqref{inq:unit-price-tight-prob-est} into Inequality \eqref{inq:unit-price-relaxation-1}, we have

    \begin{align*}
        \Ex_{\vecb\sim\sigma}[p^{unit}(\vecb)]&\leq (1-s_{i^\star})\cdot\left((1-\frac{\gamma}{\low{b}_{i^\star}-c_{i^\star}+1})\cdot(\low{b}_{i^\star}+1)+\gamma\cdot\sum_{\alpha=\low{b}_{i^\star}+2}^{\high{b}_{i^\star}-1}\left(\frac{\alpha}{\alpha-c_{i^\star}-1}-\frac{\alpha}{\alpha-c_{i^\star}}\right)\right)\\
        &\quad+\left(s_{i^\star}+(1-s_{i^\star})\cdot\gamma\cdot\frac{1}{\high{b}_{i^\star}-c_{i^\star}-1}\right)\cdot \high{b}_{i^\star}\\
        &=(1-s_{i^\star})\cdot\left(\low{b}_{i^\star}+1-\gamma\cdot\frac{\high{b}_{i^\star}-1}{\high{b}_{i^\star}-c_{i^\star}-1}+\gamma\cdot\sum_{\alpha=\low{b}_{i^\star}+2}^{\high{b}_{i^\star}-1}\frac{1}{\alpha-c_{i^\star}-1}\right)\\
        &\quad+\left(s_{i^\star}+(1-s_{i^\star})\cdot\gamma\cdot\frac{1}{\high{b}_{i^\star}-c_{i^\star}-1}\right)\cdot \high{b}_{i^\star}\\
        &=(1-s_{i^\star})\cdot\left(\low{b}_{i^\star}+1+\gamma\cdot\sum_{\alpha=\low{b}_{i^\star}+2}^{\high{b}_{i^\star}}\frac{1}{\alpha-c_{i^\star}-1}\right)+s_{i^\star}\cdot\high{b}_{i^\star}.
    \end{align*}
    The equalities simply follow from algebra.
     When $\alpha\geq \low{b}_{i^\star}+2$, we have $s_{i^\star}\cdot(\alpha-c_{i^\star}-1)\geq s_{i^\star}\cdot(\low{b}_{i^\star}+1-c_{i^\star})\geq\max_{b_{i^\star}\in[M]}U_{i^\star}(b_{i^\star},\vecb_{-i^\star})+s_{i^\star}$ by definition of $\low{b}_{i^\star}$. By dividing with $s_{i^\star}$ we get that $\alpha - c_{i^\star} - 1\geq \gamma + 1$ and thus  $\gamma/(\alpha-c_{i^\star}-1)\leq 1$. We further notice that the inequality is strict when $\alpha> \low{b}_{i^\star}+2$. Thus,
    \begin{align*}
        \Ex_{\vecb\sim\sigma}[p^{unit}(\vecb)]&< (1-s_{i^\star})\cdot(\low{b}_{i^\star}+1+(\high{b}_{i^\star}-(\low{b}_{i^\star}+2)+1)\cdot 1) +s_{i^\star}\cdot\high{b}_{i^\star}=\high{b}_{i^\star}.
    \end{align*}
    \qed
\end{proof}
\section{Proof of Theorem~\ref{t:VCG}}\label{app:VCG}

\TVCG*
\begin{proof}
\noindent Example \ref{example:VCG-can-be-bad} demonstrates an energy market where the unit price given by VCG can be arbitrarily bad compared to the unit price given by PC or PB. 
\begin{example}
    \label{example:VCG-can-be-bad}
    Let $\delta>2$ be a constant integer. Consider an energy market with $n=2k-1$ producers and $M=\delta\cdot k-2$, where each producer's supply and cost can be described as follows
    \begin{align*}
        s_i&=\begin{cases}
            \frac{1}{k}, & i\leq k,\\
            \frac{1}{2k}, & i=k+1,\\
            \frac{1}{k(i-k)(i-k+1)}, &k+2\leq i\leq 2k-1,
        \end{cases}\\
        c_i&=\begin{cases}
            0, & i\leq k,\\
            \delta, & i=k+1,\\
            \delta \cdot(i-k)-2, &k+2\leq i\leq 2k-1.
        \end{cases}
    \end{align*}
\end{example}
Recall that VCG is a truthful mechanism, and the externality determines the payment each agent receives their presence imposes on the total cost of the market, i.e., what the market would have to cost if that particular agent had not participated. Thus, the payment to each agent $i \in [k]$ in the equilibrium where everyone bids truthfully is 
    \begin{align*}
        p_i(\mathbf{c})&=\sum_{j=k+1}^{2k-1}c_j\cdot s_j+M\cdot\left(1-\left(\sum_{j=1}^{2k-1}s_j-s_i\right)\right)\\
        &=\delta \cdot \frac{1}{2k}+\sum_{j=2}^{k-1}\frac{\delta\cdot j-2}{k\cdot j\cdot (j+1)}+\frac{\delta\cdot k-2}{k^2}\\
        &=\frac{\delta}{k}\cdot \mathcal{H}_k-\frac{1}{k},
    \end{align*}
where $\mathcal{H}_k$ is the $k$-th harmonic number and the equalities hold by simple algebra. So, the worst-case unit price given by NE in VCG is at least $\Omega(\log n)$.

\noindent Next, we consider the worst-case unit price given by PB and PC. Due to Theorem \ref{t:PB}, the unit price in worst-case mixed NE in PB is bounded by $\max_{i\preceq_\mathbf{c}\tau(\mathbf{c})} \high{b}_i$. Since we have $ (\delta\cdot j-1)\cdot \frac{1}{k}\cdot(1-\sum_{\ell=2}^{j}\frac{1}{\ell\cdot(\ell-1)})=(\delta-\frac{\delta}{j})\cdot\frac{1}{k}<\delta\cdot\frac{1}{k}$ for any $j\geq 2$, the best response of any agent $i\preceq_\mathbf{c}\tau(\mathbf{c})$ when everybody else bid approximately truthful is always bidding $\delta$ or $\delta+1$. Thus, the unit price given by worst-case NE is at most $\max_{i\preceq_\mathbf{c}\tau(\mathbf{c})}\high{b}_i=\delta+1=O(1)$. The unit price in the worst-case mixed NE in PC is a bit tricky since we have not provided any upper bound for it yet. But we claim it is also at most $\max_{i\preceq_\mathbf{c}\tau(\mathbf{c})} \high{b}_i$ by a similar argument we used to prove the upper bound in PB, since the utility for the pivotal agent is the same in PB and PC. 
\qed
\end{proof}


\section{Best-case PB vs Best-Case PC}

\label{app:bestPC}

We have proved in Theorem \ref{t:1} that PC cannot be strictly dominated by any mechanism with individual rationality, including PB. However, our proof does not exclude the possibility that the best unit price by PB is always lower than the best unit price by PC, which is another kind of dominance. We now present Example \ref{example:best-PC-can-be-better} to show it is not the case. 

\label{app:b}
\begin{example}
    \label{example:best-PC-can-be-better}
    Let $\delta$ be a positive constant integer. Consider an energy market with $M=3\delta$ and $n = 3$ producers with  supplies $s_1=1/2,s_2=3/4,s_3=1/4$ and marginal costs $c_1=0,c_2=0,c_3=\delta$. 
\end{example}

\noindent  We first notice that the bidding profile $\vecb=(\delta,0,\delta)$ is a mixed NE in PC. And it provides the lowest price in PC due to Theorem \ref{thm:LB-PC}. Next, we will show that the price given by the best-case NE in PB is strictly greater than $\delta$. In this instance, we have $1\preceq_\mathbf{c}\tau(\mathbf{c})$ and $2\preceq_\mathbf{c}\tau(\mathbf{c})$. And $\low{b}_2=\delta$, while $\low{b}_1=\delta/2$. Due to a similar argument in theorem \ref{thm:LB-PB}, all agents bid at least $\max\set{\low{b}_1,\low{b}_2}$ in any mixed NE of PB. We only need to show that any bidding profile $\vecb$ with $b_1=\delta$ and $b_2=\delta$ is not an NE. It is clear that $U_2(\vecb)=\delta/2<3\delta/4\leq \max_{b'_2\in[M]}U_2(b'_2=3\delta,\mathbf{c}_{-2})$. Due to Lemma \ref{lem:least-utility-PB}, agent $2$ can increase the utility by bidding $\high{b}_2=3\delta$ regardless of the strategy of agent $3$. Therefore, $\vecb$ cannot be an NE.

\section{Further Experimental Evaluations}
\label{a:exp}
\noindent In this section, we present additional experimental evaluations of the Hedge algorithm in various typical instances.
\begin{figure}[htbp]
  \centering
  \begin{minipage}[t]{0.48\textwidth}
    \centering
    \includegraphics[height=4cm]{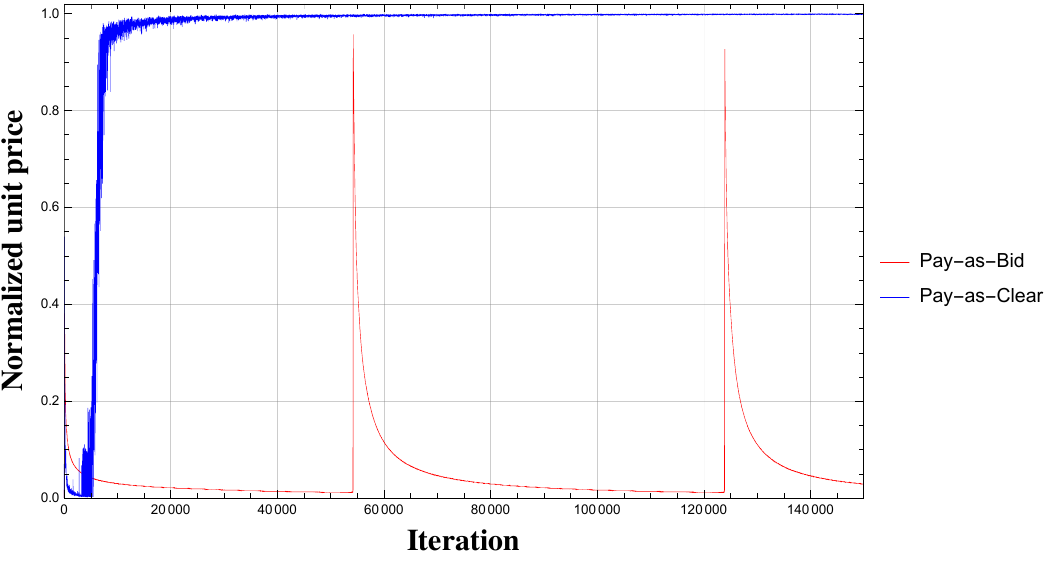}
    \caption{The normalized unit price changing by iteration in an energy market with $M=800$ and $n = 2$ producers with supplies $s_1=s_2=0.99$ and marginal costs $c_1=c_2=0$. For every agent, we have $\high{b}_i=800$ and $\low{b}_i=9$. 
    }
    \label{fig:s2agent}
  \end{minipage}
  \hfill
  \begin{minipage}[t]{0.48\textwidth}
    \centering
    \includegraphics[height=4cm]{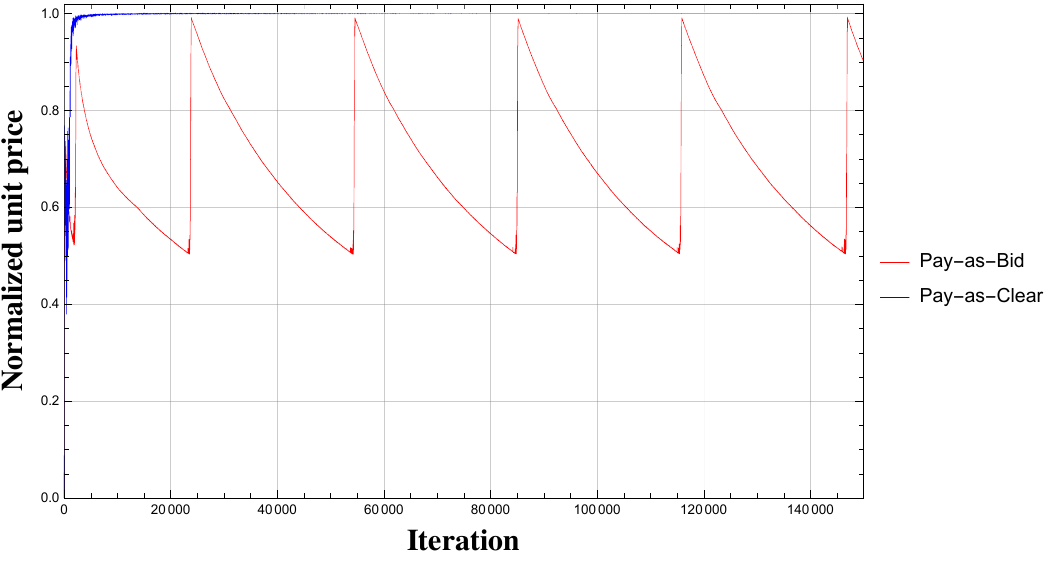}
    \caption{The normalized unit price changing by iteration in an energy market with $M=800$ and $n = 3$ producers with supplies $s_1=s_2=s_3=0.4$ and marginal costs $c_1=c_2=c_3=0$. For every agent, we have $\high{b}_i=800$ and $\low{b}_i=400$.}
    \label{fig:s3agent}
  \end{minipage}
\end{figure}
\begin{figure}[h]
  \centering
  \begin{minipage}[t]{0.48\textwidth}
    \centering
    \includegraphics[height=4cm]{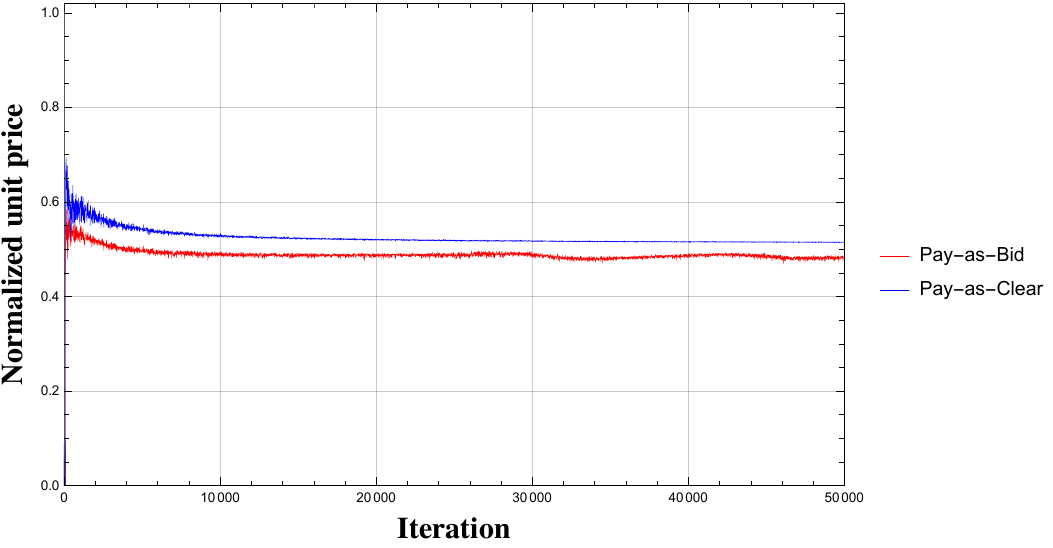}
    \caption{The normalized unit price changing by iteration in the energy market with $M=1000$ and $n = 5$ producers with randomly sampled supplies $s_1=0.72,s_2=0.15,s_3=0.47,s_4=0.96, s_5=26$ and marginal costs $c_1=390,c_2=280,c_3=30,c_4=510,c_5=680$.
    The cyclical pattern in PB is not so clear since $\max_{i\preceq_\mathbf{c}\tau(\mathbf{c})} \high{b}_i=511$ and $\max_{i\preceq_\mathbf{c}\tau(\mathbf{c})} \low{b}_i=453$ is relatively close.
    }
    \label{fig:s5agent-random}
  \end{minipage}
  \hfill
  \begin{minipage}[t]{0.48\textwidth}
    \centering
    \includegraphics[height=4cm]{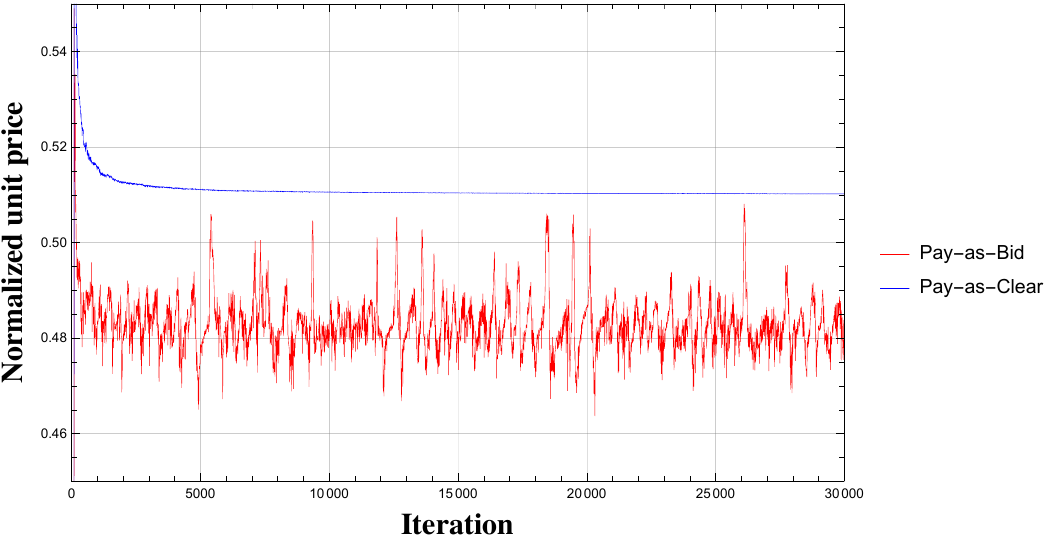}
    \caption{The normalized unit price changing by iteration in the energy market in Figure \ref{fig:s5agent-random} with scaled strategy space $M=10000$ and $n = 5$ producers with randomly sampled supplies $s_1=0.72,s_2=0.15,s_3=0.47,s_4=0.96, s_5=26$ and marginal costs $c_1=3900,c_2=2800,c_3=300,c_4=5100,c_5=6800$. We have $\max_{i\preceq_\mathbf{c}\tau(\mathbf{c})} \high{b}_i=5101$ and $\max_{i\preceq_\mathbf{c}\tau(\mathbf{c})} \low{b}_i=4521$. The figure zooms into the interval between $4500$ and $5500$.}
    \label{fig:s5agent-random-zoomed-in}
  \end{minipage}
\end{figure}

\begin{figure}[htbp]
\begin{minipage}[t]{0.48\textwidth}
    \centering
    \includegraphics[height=4cm]{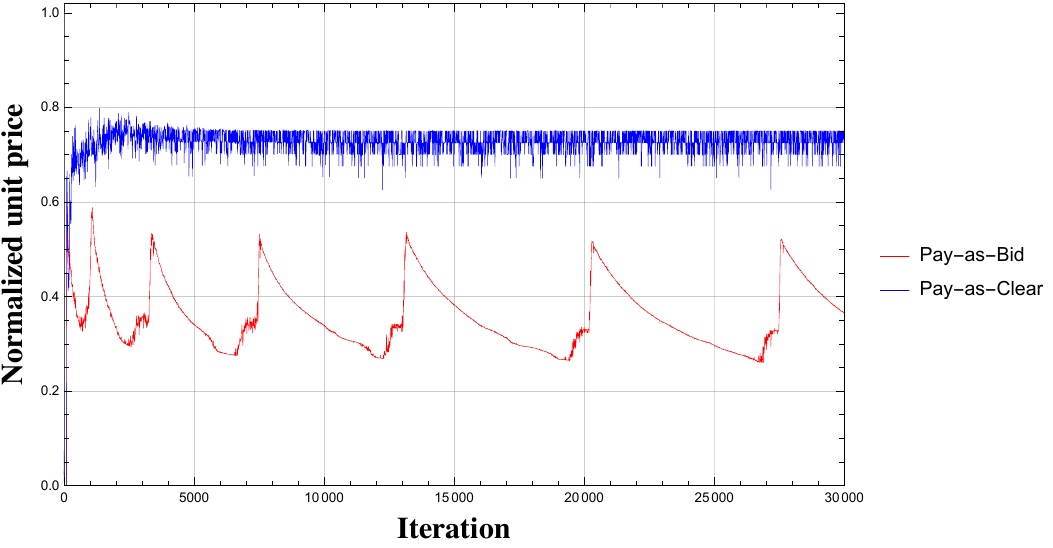}
    \caption{The normalized unit price changing by iteration in the energy market with $M=800$ and $n = 4$ producers with supplies $s_1=0.75,s_2=0.75,s_3=0.1,s_4=0.05$ and marginal costs $c_1=0,c_2=100,c_3=400,c_4=600$. We have that $\max_{i\preceq_{\mathbf{c}}\tau(\mathbf{c})}\high{b}_i=601$ and $\max_{i\preceq_{\mathbf{c}}\tau(\mathbf{c})}\low{b}_i=134$. The price in PC is volatile since the market converges to a mixed NE. In this instance, agent $2$ has two best-response leads to the same utility, and both of them are higher than $\max_{i\preceq_{\mathbf{c}}\tau(\mathbf{c})}\low{b}_i$.}
    \label{fig:a4agent}
  \end{minipage}
  \hfill
  \begin{minipage}[t]{0.48\textwidth}
    \centering
    \includegraphics[height=4cm]{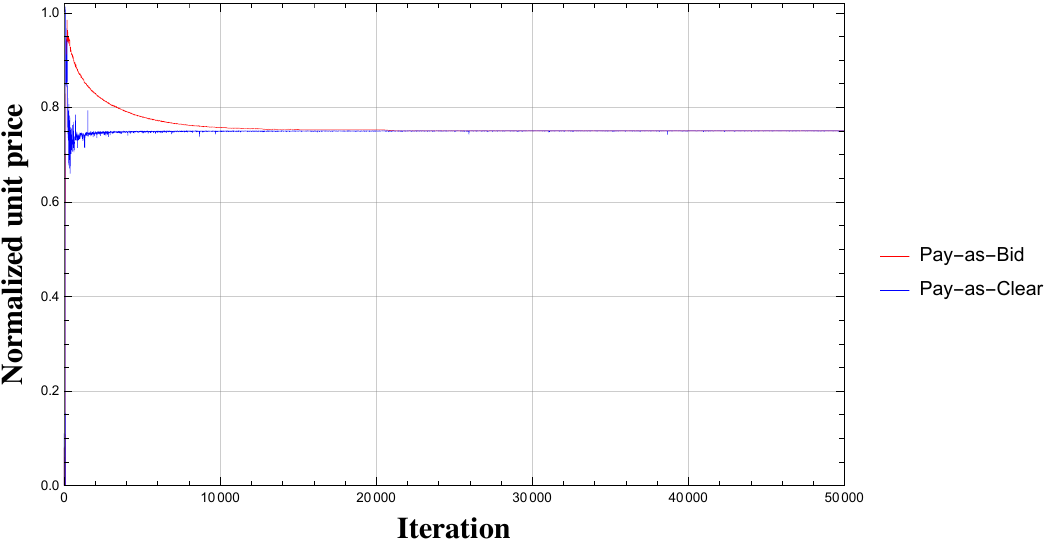}
    \caption{The normalized unit price changing by iteration in the energy market with $M=1000$ and $n = 5$ producers with supplies $s_1=s_2=s_3=s_4=0.25,s_5=0.11$ and marginal costs $c_1=0=c_2=c_3=c_4=0,c_5=600$. We have that $\max_{i\preceq_{\mathbf{c}}\tau(\mathbf{c})}\high{b}_i=601$ and $\max_{i\preceq_{\mathbf{c}}\tau(\mathbf{c})}\low{b}_i=600$. This instance is a degenerate example where $\max_{i\preceq_{\mathbf{c}}\tau(\mathbf{c})}\low{b}_i+1=\max_{i\preceq_{\mathbf{c}}\tau(\mathbf{c})}\high{b}_i$. In such degenerate instances, the unit price provided by worst NE in PB and PC is the same.}
    \label{fig:degerate-instance}
  \end{minipage}
\end{figure}

\end{document}